\documentclass{theoretics}
\usepackage{thm-restate} 
\usepackage{tikz}
\usepackage[xcolor,hyperref,notion,quotation,paper]{knowledge}
\usepackage{xspace}
\usepackage{mathrsfs}

\definecolor{ThCSdarkblue}{RGB}{25,34,64}
\knowledgestyle*{intro notion}{emphasize}
\knowledgestyle*{notion}{color=black} 

\newcommand{\RestateRemark}[1]{{\normalfont\bfseries #1}}
\newcommand{\RestateInit}[1]{\newcommand{#1}{}}
\newcommand{\RestateGo}[1]{\renewcommand{#1}{(Restated)}}

\definecolor{darkblue}{RGB}{0,0,92}
\definecolor{nicecyan}{HTML}{006165}
\definecolor{nicered}{HTML}{DB3A34}
\definecolor{nicegreen}{HTML}{6D972E}

\usetikzlibrary{calc}
\usetikzlibrary{patterns}
\usetikzlibrary{decorations}
\usetikzlibrary{decorations.markings}
\usetikzlibrary{decorations.pathmorphing}
\usetikzlibrary{decorations.pathreplacing}
\usetikzlibrary{backgrounds}
\usetikzlibrary{arrows}
\usetikzlibrary{arrows.meta}
\usetikzlibrary{bending}

\tikzset{every node/.style={font=\vphantom{Ag}, minimum size=1mm, inner sep=0pt, outer sep=0pt}}
\tikzset{every path/.style={line width=0.5pt, shorten >=3pt, shorten <=3pt}} 

\tikzstyle{dot}=[draw, shape=circle, fill, style={font=\vphantom{}}]
\tikzstyle{bluebox}=[shape=rectangle, fill=cyan!50, minimum size=4mm, inner sep=2pt, outer sep=0pt]
\tikzstyle{ubrace} = [draw, thick, decoration={brace, mirror, raise=0.0cm}, decorate, 
                      every node/.style={anchor=north, yshift=-0.1cm}]

\tikzstyle{outline} = [preaction={-, draw, shorten >=2pt, shorten <=2pt, line width=2pt, white}]

\tikzstyle{arrow} = [arrows = {->[bend,round,scale=1.5]}]
\tikzstyle{directededge} = [arrows = {-latex'}]
\tikzstyle{over} = [preaction={-, draw, shorten >=1pt, shorten <=1pt, line width=5pt, white}]
\tikzstyle{partialarrow} = [arrows = {->[bend,round,scale=1.5,harpoon]}]
\tikzstyle{surjectivearrow} = [arrows = {-To[bend,round,scale=1.5]To[bend,round,scale=1.5]}]
\tikzstyle{injectivearrow} = [arrows = {>[bend,round,scale=1.5]->[bend,round,scale=1.5]}]
\tikzstyle{mapsto} = [|->]

\tikzstyle{loop above} = [loop, looseness=15, out=50, in=130]
\tikzstyle{loop below} = [loop, looseness=15, out=230, in=310]
\tikzstyle{loop left} = [loop, looseness=15, out=140, in=220]
\tikzstyle{loop right} = [loop, looseness=15, out=-40, in=40]

\NewDocumentCommand{\mygrid}{}{
	\begin{pgfonlayer}{background}
		\draw [help lines, red!5, step=0.25] (current bounding box.south west) grid (current bounding box.north east);
		\draw [help lines, red!25, step=1] (current bounding box.south west) grid (current bounding box.north east);
		\node [help lines, red, dot] (0,0) {};
	\end{pgfonlayer}
}

\NewDocumentCommand{\myclip}{s m m}{
	\tikzset{every path/.style={}}
	\clip (#2) rectangle (#3);
	\IfBooleanTF{#1}{
		\draw [help lines, red] (#2) rectangle (#3);
		\mygrid
	}{}
	\tikzset{every path/.style={line width=0.5pt, shorten >=3pt, shorten <=3pt}}
}

\newcommand{\sst}{T}

\newcommand{\var}{\mathcal{X}}

\newcommand{\alp}{\Sigma}

\NewDocumentCommand{\out}{s m}{\IfBooleanTF{#1}{
	 \operatorname{out}(#2)}{
	"\operatorname{out}(#2)@output"}}

\newcommand{\finalupd}{O}

\let\oldotimes\otimes
\RenewDocumentCommand{\otimes}{s}{\IfBooleanTF{#1}{
	 \mathbin{\oldotimes}}{
	 \mathbin{"\oldotimes @convolution"}}}

\NewDocumentCommand{\inp}{s m}{\IfBooleanTF{#1}{
	 \operatorname{in}(#2)}{
	 "\operatorname{in}(#2) @input"}}

\NewDocumentCommand{\rel}{s m}{\IfBooleanTF{#1}{
	 \mathscr{R}(#2)}{
	 "\mathscr{R}(#2) @relation"}}

\NewDocumentCommand{\rt}{s m}{\IfBooleanTF{#1}{
	 \operatorname{root}(#2)}{
	"\operatorname{root}(#2) @root"}}

\NewDocumentCommand{\Cover}{s O{C,D} m}{\IfBooleanTF{#1}{
	 \operatorname{Cover}_{#2}(#3)}{
	"\operatorname{Cover}_{#2}(#3) @cover"}}

\NewDocumentCommand{\lrun}{s O{P} m}{\IfBooleanTF{#1}{
	 \operatorname{lft}\nolimits_{#2}^{#3}}{
	"\operatorname{lft}\nolimits_{#2}^{#3} @left run"}}

\NewDocumentCommand{\mrun}{s O{P} m}{\IfBooleanTF{#1}{
	 \operatorname{mid}\nolimits_{#2}^{#3}}{
	"\operatorname{mid}\nolimits_{#2}^{#3} @mid run"}}

\NewDocumentCommand{\rrun}{s O{P} m}{\IfBooleanTF{#1}{
	 \operatorname{rgt}\nolimits_{#2}^{#3}}{
	"\operatorname{rgt}\nolimits_{#2}^{#3} @right run"}}

\NewDocumentCommand{\wrun}{s O{P} m}{\IfBooleanTF{#1}{
	 \operatorname{run}\nolimits_{#2}({#3})}{
	"\operatorname{run}\nolimits_{#2}({#3}) @W-run"}}

\NewDocumentCommand{\seqprj}{s m}{\IfBooleanTF{#1}{
	 #2^\star}{
	"#2^\star @unmarked sequence"}}

\NewDocumentCommand{\pairprj}{s m}{\IfBooleanTF{#1}{
	 #2^\star}{
	"#2^\star @pair norm"}}

\NewDocumentCommand{\flow}{s m}{\IfBooleanTF{#1}{
	 \operatorname{flow}(#2)}{
	"\operatorname{flow}(#2)@flow"}}

\NewDocumentCommand{\pump}{s D<>{2} O{L} m}{\IfBooleanTF{#1}{
	 \operatorname{pump}^{#2}_{#3}(#4)}{
  	"\operatorname{pump}^{#2}_{#3}(#4) @pump"}}

\NewDocumentCommand{\formal}{m}{\mathtt{#1}}

\NewDocumentCommand{\sol}{s m}{\IfBooleanTF{#1}{
	 \operatorname{Sols}(#2)}{
	"\operatorname{Sols}(#2)@solutions"}}

\NewDocumentCommand{\plex}{s m}{\IfBooleanTF{#1}{
	 \le_{#2}}{
	\mathrel{"\le_{#2}@variable order"}}}
	
\NewDocumentCommand{\pgex}{s m}{\IfBooleanTF{#1}{
	 \ge_{#2}}{
	\mathrel{"\ge_{#2}@variable order"}}}
	
\NewDocumentCommand{\pleY}{s m}{\IfBooleanTF{#1}{
	 \le_{#2}}{
	\mathrel{"\le_{#2}@variant of pointwise order"}}}
	
\NewDocumentCommand{\pgeY}{s m}{\IfBooleanTF{#1}{
	 \ge_{#2}}{
	\mathrel{"\ge_{#2}@variant of pointwise order"}}}
	
\NewDocumentCommand{\ple}{s}{\IfBooleanTF{#1}{
	 \le}{
	\mathrel{"\le@pointwise order"}}}
	
\NewDocumentCommand{\pge}{s}{\IfBooleanTF{#1}{
	 \ge}{
	\mathrel{"\ge@pointwise order"}}}

\NewDocumentCommand{\pdot}{s}{\IfBooleanTF{#1}{
	 \cdot}{
	\mathbin{"\cdot@pointwise product"}}}

\NewDocumentCommand{\pdotY}{s m}{\IfBooleanTF{#1}{
	 \cdot_{#2}}{
	\mathbin{"\cdot_{#2}@variant of pointwise product"}}}

\NewDocumentCommand{\weight}{s O{u} O{v} m}{\IfBooleanTF{#1}{
	 \operatorname{weight}^{#2}_{#3}(#4)}{
  	"\operatorname{weight}^{#2}_{#3}(#4)@weight"}}

\NewDocumentCommand{\period}{s m}{\IfBooleanTF{#1}{
	 \operatorname{period}(#2)}{
  	"\operatorname{period}(#2)@period"}}

\NewDocumentCommand{\cuts}{s D<>{C} m}{\IfBooleanTF{#1}{
	 {#2}\mspace{-3mu}\operatorname{-cuts}(#3)}{
  	"{#2}\mspace{-3mu}\operatorname{-cuts}(#3)@set of cuts"}}

\NewDocumentCommand{\delay}{s D<>{C} m}{\IfBooleanTF{#1}{
	 {#2}\mspace{-3mu}\operatorname{-delay}(#3)}{
  	"{#2}\mspace{-3mu}\operatorname{-delay}(#3)@delay"}}

\let\oldiota\iota
\RenewDocumentCommand{\iota}{s}{\IfBooleanTF{#1}{
	 \oldiota}{
  	"\oldiota @initial update"}}

\NewDocumentCommand{\Sep}{s O{C,D} m}{\IfBooleanTF{#1}{
	 \operatorname{Sep}_{#2}(#3)}{
	"\operatorname{Sep}_{#2}(#3) @separation"}}

\knowledge{notion}
	| generalized sequential machine
	| generalized sequential machines

\knowledge{notion}
	| one-way transducer
	| one-way transducers
	| one-way

\knowledge{notion}
	| two-way transducer
	| two-way transducers
	| two-way
	| two-wayness

\knowledge{notion}
	| MSO-definable word transduction
	| MSO-definable word transductions
	| MSO transduction
	| MSO transductions

\knowledge{notion}
	| non-deterministic MSO transduction
	| non-deterministic MSO transductions

\knowledge{notion}
	| skeleton
	| skeletons

\knowledge{notion}
	| skeleton monoid

\knowledge{notion}
	| skeleton-idempotent
	| idempotent
	| idempotency
	| skeleton idempotency

\knowledge{notion}
	| loop
	| loops

\knowledge{notion}
	| empty@loop

\knowledge{notion}
    | pump
    | pumped
    | pumping
    
\knowledge{notion}
	| flow-normalized

\knowledge{notion}
    | period
    | periods
    
\knowledge{notion}
	| cut
	| cuts
    | $C$-cut
    | $C$-cuts
    | first $C$-cut
    | $i$-th $C$-cut
    | $(i+1)$-th $C$-cut

\knowledge{notion}
    | set of $C$-cuts
    | set of cuts

\knowledge{notion}
    | weight
    | weights

\knowledge{notion}
    | delay
    | delays

\knowledge{notion}
    | lag
    | lags

\knowledge{notion}
    | root
    | roots
    | primitive root
    | primitive roots

\knowledge{notion}
    | word inequality
    | word inequalities
    | inequality
    | inequalities

\knowledge{notion}
	| system
	| systems
    | system of inequalities
    | systems of inequalities
    | system of word inequalities
    | systems of word inequalities

\knowledge{notion}
    | solution
    | solutions

\knowledge{notion}
    | satisfiable
    | satisfiability
    | satisfied

\knowledge{notion}
    | variable order

\knowledge{notion}
    | pointwise order

\knowledge{notion}
    | pointwise product
    | pointwise products

\knowledge{notion}
    | variant of pointwise order

\knowledge{notion}
    | variant of pointwise product

\knowledge{notion}
	| convolution
	| convolutions

\knowledge{notion}
	| length-preserving

\knowledge{notion}
	| automatic

\knowledge{notion}
	| rational relation
	| rational relations
	| rational

\knowledge{notion}
	| $k$-valued
	| $k$-valuedness
	| $2$-valued
	| $2^n$-valued
    | $2^k$-valued
	| valuedness
	| single-valued
	| single-valuedness
	| functional

\knowledge{notion}
	| finite-valued
	| Finite-valued
	| finite valuedness
	| Finite valuedness

\knowledge{notion}
	| copyless update
	| copyless updates
        | copyless
	| update
	| updates

\knowledge{notion}
	| initial@update
	| initial update

\knowledge{notion}
	| final@update
	| final update

\knowledge{notion}
	| deterministic
	| determinism
	| Deterministic

\knowledge{notion}
	| unambiguous
	| unambiguity

\knowledge{notion}
	| $k$-ambiguous
	| $k$-ambiguity
	| Ambiguity
	| ambiguity
	| finite ambiguity
	| Finite ambiguity
        | finite-ambiguous
        | Finite-ambiguous

\knowledge{notion}
	| streaming string transducer
	| streaming string transducers
	| SST
	| SSTs

\knowledge{notion}
	| run
	| runs

\knowledge{notion}
	| input consumed by
	| inputs consumed by 
	| consumed input
	| consumed inputs
	| input
	| inputs

\knowledge{notion}
	| update induced by
	| induced update

\knowledge{notion}
	| accepting

\knowledge{notion}
	| output produced by
	| outputs produced by
	| produced output
	| produced outputs
	| output
	| outputs

\knowledge{notion}
	| relation realized by
	| relations realized by
	| realized relation
	| realized relations
	| relation
	| relations

\knowledge{notion}
	| equivalent
	| equivalence
	| equivalences
	| equivalence problem
	| Equivalence

\knowledge{notion}
    | dumbbell from $q_1$ to $q_2$
    | dumbbell
    | dumbbells

\knowledge{notion}
    | plain dumbbell from $q_1$ to $q_2$
    | plain dumbbell
    | plain dumbbells

\knowledge{notion}
    | multiset finite-state automaton
    | multiset finite-state automata
    | multiset automaton
    | multiset automata

\knowledge{notion}
	| W-pattern
	| W-patterns

\knowledge{notion}
    | left run
    | left runs

\knowledge{notion}
    | mid run
    | mid runs

\knowledge{notion}
    | right run
    | right runs

\knowledge{notion}
    | W-run
    | W-runs

\knowledge{notion}
	| divergent
	| divergence

\knowledge{notion}
	| simply divergent
	| simple divergence

\knowledge{notion}
	| marked sequence
	| marked sequences

\knowledge{notion}
    | unmarked sequence

\knowledge{notion}
	| cover

\knowledge{notion}
	| reversible

\knowledge{notion}
	| size
	| sizes

\knowledge{notion}
	| separation

\ThCSauthor[1]{Emmanuel Filiot}{emmanuel.filiot@ulb.be}[0000-0002-2520-5630]
\ThCSauthor[2]{Isma\"{e}l Jecker}{ismael.jecker@femto-st.fr}[0000-0002-6527-4470]
\ThCSauthor[3]{Christof L\"{o}ding}{loeding@automata.rwth-aachen.de}[0000-0002-1529-2806]
\ThCSauthor[4]{Anca Muscholl}{anca@labri.fr}[0000-0002-8214-204X]
\ThCSauthor[5]{Gabriele Puppis}{gabriele.puppis@uniud.it}[0000-0001-9831-3264]
\ThCSauthor[6]{Sarah Winter}{sarah.winter@irif.fr}[0000-0002-3499-1995]

\ThCSaffil[1]{Universit\'{e} libre de Bruxelles, Belgium}
\ThCSaffil[2]{Universit\'{e} de Besan\c{c}on, France}
\ThCSaffil[3]{RWTH Aachen University, Germany}
\ThCSaffil[4]{LaBRI, Universit\'{e} Bordeaux, France}
\ThCSaffil[5]{University of Udine, Italy}
\ThCSaffil[6]{IRIF, Universit\'{e} Paris Cit\'{e}, France}

\ThCSshortnames{E.~Filiot, I.~Jecker, C.~L\"{o}ding, A.~Muscholl, G.~Puppis, S.~Winter}

\title{Finite-valued Streaming String Transducers}
\ThCSshorttitle{Finite-valued Streaming String Transducers}
\ThCSthanks{Part of this work was initiated during the Dagstuhl seminar 23202 ``Regular Transformations''~\cite{DagRep.13.5.96}. A preliminary version of this article appeared  at LICS 2024~\cite{FiliotJLMPW24}.}
\ThCSkeywords{streaming string transducers, finite valuedness}
          
\ThCSyear{2025}
\ThCSarticlenum{1}
\ThCSreceived{Jun 10, 2024}
\ThCSaccepted{Nov 10, 2024}
\ThCSpublished{Jan 08, 2025}
\ThCSdoicreatedtrue

\begin{document}

\maketitle

\begin{abstract}
       A  transducer is finite-valued if there exists a bound
    $k$ such that any given
    input maps to at most $k$ outputs. For classical
    one-way transducers, it is known since the 1980s that
    finite valuedness entails decidability of the equivalence
    problem. This result is in contrast with undecidability in the
    general case, making finite-valued transducers
    very appealing. For one-way transducers it is also known that
    finite valuedness itself is decidable and that any
    $k$-valued finite transducer can be decomposed into a union of $k$
    single-valued finite transducers.

    In this article, we extend the above results to copyless streaming string
    transducers (SSTs), addressing open questions raised by Alur and
    Deshmukh in 2011. SSTs strictly extend the expressiveness of
    one-way transducers via additional variables that store
    partial outputs. We prove that any $k$-valued SST can be
    effectively decomposed into a union of $k$ (single-valued)
    deterministic SSTs. As a corollary, we establish the equivalence
    between SSTs and two-way transducers in the finite-valued case, even though these models are generally incomparable. Another corollary provides an elementary upper bound for
    deciding equivalence of finite-valued SSTs. The equivalence problem was already
    known to be decidable, but the proof complexity relied
    on Ehrenfeucht's
    conjecture. Lastly, our main contribution shows that finite valuedness of SSTs is
    decidable, with the complexity  being {\upshape\sc PSpace} in general, and  {\upshape\sc PTime} when the
    number of variables is fixed.  
\end{abstract}

\section{Introduction}\label{sec:intro}
\enlargethispage{-\baselineskip}
Finite-state word transducers are simple devices that allow 
effective reasoning about data transformations. 
In their most basic form, they transform words using finite control. 
\AP
For example, the oldest transducer model, 
known as ""generalized sequential machine"", 
extends deterministic finite state automata by associating each "input" with a corresponding "output"
that is generated by appending finite words specified along the transitions.
This rather simple model of transducer is capable of representing 
basic partial functions between words, e.g.~the left 
rotating function $a_1\: a_2\dots a_n\mapsto a_2\dots a_n \: a_1$.
\AP
Like automata, transducers can also be enhanced with non-determinism,
as well as the ability of scanning the input several times (""two-wayness"").
\AP
For example, the non-deterministic counterpart of 
"generalized sequential machines", called here ""one-way transducers"", 
can be used to represent the right rotating function 
$a_1\dots a_{n-1}\: a_n\mapsto a_n\:a_1\dots a_{n-1}$,
but also word relations that are not partial functions, like for 
instance the relation that associates an input $a_1\dots a_n$ 
with any output from $\{a_1\}^*\dots \{a_n\}^*$.
Similarly, deterministic "two-way transducers" can compute 
the mirror function $a_1\dots a_n \mapsto a_n\dots a_1$,
the squaring function $w \mapsto w w$, etc.

\AP
Inspired by a logic-based approach applicable to arbitrary relational structures~\cite{CE12}, 
""MSO-definable word transductions"" were considered by Engelfriet and
Hoogeboom~\cite{eh01} and shown to be "equivalent" to deterministic 
"two-way transducers".
\AP
Ten years later Alur and Cern{\'y}~\cite{DBLP:conf/fsttcs/AlurC10} proposed 
\reintro{streaming string transducers} (\reintro{SSTs} for short), 
a "one-way" model that uses write-only variables as additional storage. 
In "SSTs", variables store strings and  can be updated by appending or prepending strings, or 
concatenated together, but not duplicated (they are "copyless"). 
Alur and Cern{\'y} also showed that, in the "functional" case, 
that is, when restricting to transducers that represent partial functions,
"SSTs" are "equivalent" to the model studied in~\cite{eh01},
and thus in particular to "two-way transducers".
These "equivalences" between transducer models motivate nowadays 
the use of the term ``regular'' word function, in the spirit of
classical results on regular word languages from automata theory 
and logics due to B\"uchi, Elgot, Trakhtenbrot, Rabin, and others. 

While transducers inherit features like
non-determinism and "two-wayness" from automata,
these characteristics have an
impact on their expressive power compared to automata. 
In the case of automata it is known that adding non-determinism and 
"two-wayness" does not affect the expressive power,
as it only makes them more succinct in terms of number of states.
It does not affect decidability of fundamental problems either,
though succinctness makes some problems 
computationally harder. 
In contrast, for transducers, non-determinism 
and/or "two-wayness"
significantly  change the expressive power.
For instance, non-deterministic "one-way transducers" may capture 
relations that are not partial functions,
and thus not computable by "generalized sequential machines"\footnote{Even if a
non-deterministic one-way transducer computes a partial function, there may be
no equivalent deterministic one-way transducer. However, this question can be decided in {\upshape\sc PTime}~\cite{BealC02,AllauzenM03}.}.
This difference is also apparent at the level of decidability results. 
For example, the "equivalence problem" is in {\upshape\sc NLogSpace}
for "generalized sequential machines",
and undecidable for "one-way transducers" \cite{FischerR68,iba78siam,many_facets}.
We should also mention that in the "functional" case it is 
possible to convert one transducer model to another
(e.g.~convert an "SST" to an "equivalent" "two-way transducer").
In the non-"functional" case, the picture is less satisfactory.
In particular, non-deterministic "SSTs" and non-deterministic 
"two-way transducers" turn out to be incomparable:
for example, the relation
$\{(u\,v, \, v\,u) \::\: u,v\in\Sigma^*\}$ can be represented 
in the former model but not in the latter, while the relation
$\{(w, \, w^n) \::\: w\in\Sigma^*, n\in\mathbb{N}\}$ can be
represented in the latter model but not in the former.
\AP
However, "SSTs" can still be converted to "equivalent"
""non-deterministic MSO transductions""~\cite{DBLP:conf/icalp/AlurD11}, 
which extend the original "MSO transductions" 
by existentially quantified monadic parameters.

\AP
There is however a class of relations that is close
to (regular) word functions in terms of good behavior: 
the class of \reintro{finite-valued} relations.
These are relations that associate a uniformly bounded
number of "outputs" with each "input".
\AP
The concept of bounding the number of outputs associated 
with each "input" in a transducer is closely related to 
the notion of  \reintro{finite ambiguity}, which refers to bounding the number of "accepting" "runs".
 "Ambiguity" has been intensively studied 
in the context of formal languages, where it is shown,
for instance, that "equivalence" of "unambiguous"
automata is decidable in {\upshape\sc PTime} \cite{Stearns1985OnTE}.
In the context of relations, "$k$-valuedness" and "$k$-ambiguity"
were initially considered in the setting of "one-way transducers". 
For example, \cite{GI83} showed that, for fixed $k$,
one can decide in {\upshape\sc PTime} whether a given "one-way transducer"
is "$k$-valued". Similarly, "$k$-valuedness" for fixed $k$ can be decided
in {\upshape\sc PSpace} for "two-way transducers" and "SSTs"~\cite{DBLP:conf/icalp/AlurD11}.

It is also clear that every "$k$-ambiguous"
"one-way transducer" is "$k$-valued". 
Conversely, it was shown that every "$k$-valued"
"one-way transducer" can be converted to an 
"equivalent", "$k$-ambiguous" one 
\cite{Weber93,web96,DBLP:journals/mst/SakarovitchS10}.
This result, even if it deals with a rather simple model
of transducer, already uses advanced normalization techniques 
from automata theory and involves an exponential
blow up in the number of states, as shown in the 
example below.

\begin{example}\label{ex:finitevalued}
    Fix $k\in\mathbb{N}$ and consider the relation
    \[
      R_k ~=~ 
      \big\{ 
        (w_1\dots w_n, \, w_i) \::\:
        n\in\mathbb{N}, \: 1\leq i\leq n, \:
        w_1,\dots,w_n \in \{0,1\}^k
      \big\}.
    \]
    Examples of pairs in this relation, for $k=2$, are 
    $(00\,10\,11,\,00)$, 
    $(00\,10\,11,\,10)$, and 
    $(00\,10\,11,\,11)$. 
    \AP
    For arbitrary $k$, the relation $R_k$ 
    associates at most $2^k$ "outputs" with each "input"
    (we say it is \reintro{$2^k$-valued}).
    This relation can be realized by "one-way transducers"
    that exploit non-determinism to guess which block from
    the "input" becomes the "output". 
    For instance, a possible "transducer@SST" $T_k$ that realizes $R_k$
    repeatedly consumes blocks of $k$ bits from the input, 
    without outputting anything, until it non-deterministically 
    decides to copy the next block, and after that it 
    continues consuming the remaining blocks without output.
    Note that this transducer $T_k$ has $\mathcal{O}(k)$
    states, it is "finite-valued", but not "finite-ambiguous",
    since the number of "accepting" "runs" per "input" 
    depends on the number $n$ of blocks in the input and it 
    is thus unbounded. 
    A "finite-ambiguous" transducer realizing the same relation
    $R_k$ can be obtained at the cost of an exponential blow-up
    in the number of states, 
    for instance by initially guessing and outputting a $k$-bit
    word $w$ (this requires at least $2^k$ states), 
    and then verifying that $w$ occurs 
    as a block of the input. 
\end{example}

Since "$k$-ambiguous" automata can be easily 
decomposed into a  union of $k$ "unambiguous" automata,
the possibility of converting a "$k$-valued" "one-way transducer" 
to a "$k$-ambiguous" one   entails a decomposition result of the 
following form: every "$k$-valued" "one-way transducer" is equivalent
to a finite union of "functional" "one-way transducers".
One advantage of this type of decomposition is that it allows to generalize
the decidability of the "equivalence" problem from "functional" 
to "$k$-valued" "one-way transducers",
which brings us back to the original motivation for considering 
classes of "finite-valued" relations.
Decidability of the "equivalence problem" for
"$k$-valued" "one-way transducers" was independently established
in~\cite{ck86}.
The latter work also states that the same techniques can be adapted 
to show decidability of "equivalence" for "$k$-valued" 
"two-way transducers" as well.
Inspired by~\cite{ck86}, the "equivalence problem" was later shown 
to be decidable also for "$k$-valued" "SSTs"~\cite{equivalence-finite-valued}.
However, the decidability results from~\cite{ck86} and~\cite{equivalence-finite-valued} 
rely on the Ehrenfeucht conjecture~\cite{AL85,Guba86} and therefore provide no elementary
upper bounds on the complexity.

Decomposing "finite-valued" "SSTs" and deciding
"finite valuedness" for "SSTs" were listed as open problems 
in~\cite{DBLP:conf/icalp/AlurD11}, more than 10 years ago,
and represent our main contributions.
Compared to "one-way transducers", new challenges arise with
"SSTs", due to the extra power they enjoy to produce "outputs". 
For example, consider the relation consisting of all pairs
of the form $(w, 0^{n_0}1^{n_1})$ or $(w, 1^{n_1}0^{n_0})$, 
where $w\in\{0,1\}^*$, and $n_b$ ($b=0,1$) is the number of 
occurrences of $b$ in $w$. 
This relation is "$2$-valued", and is not realizable by any "one-way transducer". 
On the other hand, the relation is realized by an "SST" $T$ with a single state 
and two variables, denoted $X_0,X_1$ and both initially empty: 
whenever $T$ reads $b\in\{0,1\}$,
it non-deterministically applies the "update" $X_b := b \,X_b$ or 
$X_b := X_b\,b$ (while leaving $X_{1-b}$ unchanged); at the end of
the input, $T$ "outputs" either $X_0 \, X_1$ or $X_1 \, X_0$. 
The ability of "SSTs" to generate "outputs" in a non-linear way
makes their study challenging and intriguing. To illustrate this, 
consider a slight modification of $T$ where
$X_0$ is initialized with $1$, instead of the empty word:
the new "SST" is not "finite-valued" anymore, because upon 
reading $0^n$ it could "output" any word of the form 
$0^i \,1 \, 0^j$, with $i,j\in\mathbb{N}$ such that $i+j = n$.

Another open problem was to compare the expressive power of "SSTs" and 
"two-way transducers" in the "finite-valued" case.
It is  not hard to see that the standard translation from deterministic
"two-way transducers" to "deterministic" "SSTs" also applies to the
"finite-valued" case (cf.~second part of the proof of Theorem~\ref{thm:finval-SST-2FT}).
The converse translation, however, is far more complicated and relies on
the decomposition theorem for "SSTs" which we establish in this article.

\medskip
\subparagraph*{\bf Contributions.} 
The results presented in this article  draw a
rather complete picture about 
"finite-valued" "SSTs", answering several open problems from~\cite{DBLP:conf/icalp/AlurD11}.
First, we show that "$k$-valued" "SSTs" enjoy the same decomposition property as
"one-way transducers":  
\RestateInit{\restateDecomposition}
\begin{restatable}{theorem}{Decomposition}
\label{thm:FiniteValuedToDecomposition}\RestateRemark{\restateDecomposition}
    For all $k \in \mathbb{N}$, 
    every $k$-valued "SST" can be effectively decomposed into a union of
    $k$ "single-valued" (or even "deterministic") "SSTs".
    The complexity of the construction is elementary.
\end{restatable}

A first consequence of the above theorem is the "equivalence" of
"SSTs" and "two-way transducers" in the "finite-valued" setting:

\begin{theorem}\label{thm:finval-SST-2FT}
  Let $R\subseteq \alp^*\times \alp^*$ be a "finite-valued" relation.
  If $R$ can be realized by an "SST", then an "equivalent" 
  "two-way transducer" can be effectively constructed, and vice-versa.
\end{theorem}

\begin{proof} 
 If $R$ is realized by an "SST" $T$, then we can apply
  Theorem~\ref{thm:FiniteValuedToDecomposition} to obtain $k$
  "unambiguous" "SSTs" $T_1,\dots,T_k$ whose union is "equivalent"
  to $T$.
  From \cite{DBLP:conf/fsttcs/AlurC10} we know that in the functional
  case, "SSTs" and two-way transducers are "equivalent".
  Thus, every $T_i$ can be transformed effectively into an "equivalent",
  even deterministic,  two-way transducer.
  From this we obtain an "equivalent" "$k$-ambiguous" two-way transducer.

  For the converse we start with a "$k$-valued" "two-way transducer" $T$ and
  first observe that we can normalise $T$ in such a way that the
  crossing sequences%
\footnote{A crossing sequence is a standard notion
    for finite-state two-way machines~\cite{DBLP:journals/ibmrd/Shepherdson59},
    and it is defined as the sequence of
    states in which a given input position 
    is visited by an accepting run of the machine.}
  of "accepting" "runs" of $T$ are bounded by a
  constant linear in the size of $T$.
  Once we work with runs with bounded crossing sequences we can
  construct an "equivalent" "SST" in the same way as we do for
  deterministic two-way transducers.
  The idea is that during the run of the "SST" 
  the variables record the "outputs"
  generated by the pieces of runs situated to the left 
  of the current input position 
  (see e.g.~\cite{ledent13,DBLP:journals/ijfcs/DartoisJR18} for self-contained proofs). 
\end{proof}

A second consequence of Theorem~\ref{thm:FiniteValuedToDecomposition}
is an elementary upper bound for the 
"equivalence problem" of "finite-valued"
"SSTs"~\cite{equivalence-finite-valued}:

\begin{theorem}\label{thm:equiv-finval-SST}
  The "equivalence problem" for "$k$-valued" "SSTs" can be solved with
  elementary complexity.
\end{theorem}

\begin{proof}
  Given two "$k$-valued" "SSTs" $T,T'$, we first decompose them into
  unions of $k$
  deterministic "SSTs" $T_1,\dots,T_k$ and $T'_1,\dots,T'_k$,
  respectively.
  Finally, \cite[Theorem 4.4]{DBLP:conf/icalp/AlurD11} shows how to check the
  "equivalence" of $\bigcup_{i=1}^k T_i$ and $\bigcup_{i=1}^k
      T'_i$  in {\upshape\sc PSpace}.
\end{proof}

    Our last, and main, contribution establishes the decidability
    of "finite valuedness" for~"SSTs":

\RestateInit{\restateEffectiveness}
\begin{restatable}{theorem}{Effectiveness}
	\label{thm:effectiveness}\RestateRemark{\restateEffectiveness}
      Given any "SST" $T$, we can decide in {\upshape\sc PSpace}
      if $T$ is
      "finite-valued" (and if the number of variables
      is fixed then the complexity is {\upshape\sc PTime}). 
      Moreover, this problem is at least as hard as the
      "equivalence problem" for "deterministic" "SSTs". 
    \end{restatable}

    This last result is the most technical one, and requires to reason
    on particular substructures ("W-patterns") of "SSTs".
    Such substructures 
    have  been already  used for "one-way transducers", 
    but for "SSTs" genuine challenges arise.
    The starting point of our  proof is a
    recent result allowing to determine if two runs of
    an "SST" are
    far apart~\cite{SSTdelay}.
    The proof then relies on identifying suitable patterns and
    extending techniques from word combinatorics to more involved
    word inequalities. 
    

    Based on the "equivalence" between "SSTs" and "two-way transducers" in the
    "finite-valued" setting (Theorem~\ref{thm:finval-SST-2FT}), and the
    decidability of "finite valuedness" for "SST"
    (Theorem~\ref{thm:effectiveness}), we exhibit 
    an alternative proof for the following (known)
    result:

    \RestateInit{\restateYen}
\begin{restatable}[{\cite{YenY22}}]{corollary}{Yen}
	  \label{thm:finval-twoway}\RestateRemark{\restateYen}
      "Finite valuedness" of "two-way transducers" is decidable in \upshape{{\upshape\sc PSpace}}.
    \end{restatable}

    Observe also that without the results in this paper, the result
    of~\cite{YenY22}   could not help to show
    Theorem~\ref{thm:effectiveness}, because 
    only the
    conversion from "finite-valued" "two-way transducers" to "finite-valued" "SSTs" was
    known (under the assumption that any input positions is visited a
    bounded number of times),
    but not the other way around. 
    %
    Also note that Theorems~\ref{thm:FiniteValuedToDecomposition}
    and~\ref{thm:finval-SST-2FT} together
    imply a
    decomposition result for "finite-valued" "two-way transducers".

    Similar results can be derived for non-deterministic "MSO transductions". 
    More precisely, since "SSTs" and non-deterministic 
    "MSO transductions" are "equivalent"~\cite{DBLP:conf/icalp/AlurD11},
    Theorem~\ref{thm:effectiveness} entails decidability of
    "finite valuedness" for non-deterministic 
    "MSO transductions" as well. Moreover, since in the "single-valued" case, 
    "deterministic" "SSTs"
    and "MSO transductions" are "equivalent"~\cite{DBLP:conf/fsttcs/AlurC10}, 
    Theorem~\ref{thm:FiniteValuedToDecomposition} implies a decomposition result for 
    "MSO transductions": any "$k$-valued" non-deterministic 
    "MSO transduction" can be decomposed as a union of $k$ (deterministic) 
    "MSO transductions". Finally, from Theorem~\ref{thm:finval-SST-2FT}, we
    also obtain that, under the assumption of "finite valuedness", 
    "non-deterministic MSO transductions", 
    "two-way transducers", and "SSTs" are equally expressive.

\section{Preliminaries}\label{sec:preliminaries}

For convenience, technical terms and notations in the
electronic version of this manuscript are hyper-linked to their 
definitions (cf.~\url{https://ctan.org/pkg/knowledge}).

Hereafter, $\mathbb{N}$ (resp. $\mathbb{N}_+$) 
denotes the set of non-negative (resp.~strictly positive) integers,
and $\alp$ denotes a generic alphabet.

\medskip
\subparagraph*{\bf Words and relations.}
We denote by $\varepsilon$ the empty word, 
by $|u|$ the length of a word $u\in\alp^*$, and by
$u[i]$ its $i$-th letter, for $1\leq i\leq |u|$. 
We introduce a "convolution" operation on words,
which is particularly useful to identify robust and
well-behaved classes of relations, 
as it is done for instance in the theory of automatic structures \cite{BlumensathG00}.
For simplicity, we only consider "convolutions" of 
words of the same length. 
\AP
Given $u,v\in\alp^*$, with $|u|=|v|$,
the \reintro{convolution} $""u \otimes* v@convolution""$ 
is a word over $(\alp^2)^*$ of length $|u|=|v|$ such that 
$(u\otimes v)[i] = (u[i],v[i])$ for all $1\leq i\leq |u|$. 
For example, $(aba)\otimes (bcc) = (a,b)(b,c)(a,c)$. 
As $\otimes$ is associative, we may write 
$u\otimes v\otimes w$ for any words $u,v,w$. 

\AP
A relation $R\subseteq (\alp^*)^k$ is ""length-preserving""
if $|u_1|=\dots=|u_k|$ for all $(u_1,\dots,u_k)\in R$. 
\AP
A "length-preserving" relation is ""automatic""
if the language 
$\{ u_1\otimes \dots \otimes u_k \mid (u_1,\dots,u_k)\in R\}$ 
is recognized by a finite state automaton.
\AP
A binary relation $R\subseteq \alp^*\times \alp^*$ (not
necessarily "length-preserving")
is ""$k$-valued"", for $k\in\mathbb{N}$, if for all $u\in\alp^*$, 
there are at most $k$ words $v$ such that $(u,v)\in R$. 
\AP
It is ""finite-valued"" if it is "$k$-valued" for some $k$.

\medskip
\subparagraph*{\bf Variable updates.} 
Fix a finite set of variables $\var = \{X_1, \dots, X_m\}$,
disjoint from the alphabet $\alp$.
\AP
A (copyless) ""update"" is any mapping 
$\alpha: \var\rightarrow (\alp \uplus \var)^*$
such that each variable $X\in\var$ appears \emph{at most once}
in the word $\alpha(X_1)\dots\alpha(X_m)$.
Such an "update" can be morphically extended to words over $\alp\uplus\var$,
by simply letting $\alpha(a) = a$ for all $a\in\alp$.
Using this, we can compose any two "updates" $\alpha,\beta$ to form 
a new "update" $\alpha \, \beta : \var \rightarrow (\alp \uplus \var)^*$, 
defined by $(\alpha \, \beta)(X) = \beta(\alpha(X))$ for all $X\in\var$.
\AP
An "update" is called ""initial@@update"" (resp.~""final@@update"")
if all variables in $\var$ (resp.~$\var\setminus\{X_1\}$) are mapped
to the empty word. The designated variable $X_1$ is used to store
the final "output" produced by an "SST", as defined in the next
paragraph.

\medskip
\subparagraph{\bf Streaming string transducers.}
A (non-deterministic, "copyless") 
""streaming string transducer""
(\reintro{SST} for short) 
is a tuple 
$T = (\alp, \var, Q, Q_{\mathrm{init}},$  $Q_{\mathrm{final}}, \finalupd, \Delta)$,
where $\alp$ is an alphabet, $\var$ is a finite set of variables, $Q$ is a finite set of states, 
$Q_{\mathrm{init}}, Q_{\mathrm{final}} \subseteq Q$ are the sets of 
initial and final states, $\finalupd$ is a function from final states
to "final@@update" "updates", and
$\Delta$ is a finite transition relation consisting of tuples of the form 
$(q, a, \alpha, q')$, where $q,q'\in Q$ are the source and target states, 
$a\in\alp$ is an input symbol, and $\alpha$ is an "update".
We often denote a transition $(q, a, \alpha, q')\in\Delta$ by the 
annotated arrow:
\[q \xrightarrow{a/\alpha} q'.\]
\AP
The ""size"" $|T|$ of an "SST" $T$ is defined as the number 
of states plus the size of its transition relation.

\AP
A ""run"" of $\sst$ 
is a sequence of transitions from $\Delta$ of the form
\[
\rho = q_0 \xrightarrow{a_1/\alpha_1} q_1 \xrightarrow{a_2/\alpha_2} q_2 \dots
q_{n-1} \xrightarrow{a_n/\alpha_n} q_n.
\]
\AP
The \reintro{input consumed by} $\rho$ is the word
$""\inp*{\rho} @input"" = a_1\dots a_n$.
\AP
The ""update induced by"" $\rho$ is the composition $\beta = \alpha_1 \dots \alpha_n$.
We write $\rho:\: u/\beta$ 
to mean that $\rho$ is a "run" with $u$ as
"consumed input" and $\beta$ as "induced update".
\AP
A "run" $\rho$ as above is ""accepting"" if
the first state is initial and the last state is final,
namely, if $q_0\in Q_{\mathrm{init}}$ and $q_n\in Q_{\mathrm{final}}$. 
\AP
In this case, the "induced update", extended to the left
with the "initial update" denoted by $\iota$ and to the right with
the "final update" $\finalupd(q_n)$, gives rise to an
"update" $\iota \, \beta \, \finalupd(q_n)$ that maps $X_1$ to a word 
over $\alp$ and all remaining variables to the empty word.
In particular, the latter "update" determines the 
\reintro{output produced by} $\rho$,
defined as the word 
$""\out*{\rho} @output"" = (\iota \, \beta \, \finalupd(q_n))(X_1)$.

\AP
The \reintro{relation realized by} an "SST" $T$ is 
\[
  ""\rel*{T} @relation""
  ~=~ 
  \big\{ \big(\inp{\rho},\out{\rho}\big)\in \alp^*\times \alp^* ~\big|~
         \rho\text{ "accepting" "run" of }T \big\}
\]
\AP
An "SST" is \reintro{$k$-valued} (resp.~\reintro{finite-valued}) 
if its "realized relation" is so.
\AP
It is ""deterministic"" 
if it has a single
initial state and the transition relation is a partial function 
(from pairs of states and input letters to pairs of "updates" and states).
\AP
It is ""unambiguous"" if it admits at most one "accepting" "run" on each "input".
\AP
Similarly, it is called ""$k$-ambiguous"" if it admits at most 
$k$ "accepting" "runs" on each "input". 
Of course, every "deterministic" "SST" is "unambiguous", and every
"unambiguous" "SST" is "single-valued" (i.e. $1$-valued).
\AP
Two "SSTs" $\sst_1,\sst_2$ are ""equivalent"" 
if $\rel{\sst_1}=\rel{\sst_2}$. 
The "equivalence problem" for "SSTs" is undecidable in
general, and it is so even for one-way transducers~\cite{FischerR68,iba78siam}. 
However, decidability is recovered for "finite-valued" "SSTs":

\begin{theorem}[\cite{equivalence-finite-valued}]
The "equivalence" problem for "finite-valued" "SSTs" is decidable.
\end{theorem}

Note that checking equivalence is known to be in 
{\upshape\sc PSpace} for \emph{"deterministic"} "SSTs".
This easily generalizes to 
unions of "deterministic" (hence "single-valued") "SSTs", because the
"equivalence" checking algorithm is exponential only in the number of variables: 

\begin{theorem}[\cite{DBLP:conf/icalp/AlurD11}]
The following problem is in \upshape{{\upshape\sc PSpace}}: 
given $n+m$ "deterministic" "SSTs" $\sst_1,\dots,\sst_n,\sst'_1,\dots,\sst'_m$, 
decide whether $\bigcup_{i=1}^n\rel{\sst_i} = \bigcup_{j=1}^m\rel{\sst'_j}$.
\end{theorem}

For any fixed $k$, the "$k$-valuedness" property is 
decidable in {\upshape\sc PSpace}:

\begin{theorem}[\cite{DBLP:conf/icalp/AlurD11}]
For any fixed $k\in\mathbb{N}$, the following problem 
is in\footnote{In
  \cite{DBLP:conf/icalp/AlurD11}, no complexity result is provided,
  but the decidability procedure relies on a reduction to the
  emptiness of a $1$-reversal $k(k+1)$-counter machine, based on the
  proof for equivalence of deterministic
SST~\cite{DBLP:conf/popl/AlurC11}. The counter machine is 
exponential in the number of variables only. The result follows since the
emptiness problem for counter machines with fixed number of reversals and fixed
number of counters is in
{\upshape\sc NLogSpace}~\cite{DBLP:journals/jcss/GurariI81}.} \upshape{{\upshape\sc PSpace}}: 
given an "SST" $T$, decide whether $T$ is "$k$-valued".
It is in {\upshape\sc PTime} if one further restricts to "SSTs"
with a fixed number of variables.
\end{theorem}

The decidability status of "finite valuedness" for "SSTs", 
i.e., if $k$ is unknown, was an open problem.
Part of our contribution is to show that this problem is  decidable, too.

\subsection{Pumping and word combinatorics}\label{subsec:pumping}

When reasoning with automata, it is common practice
to use pumping arguments.
This section introduces pumping for "SSTs", as well 
as combinatorial results for reasoning about (in)equalities 
between pumped "outputs" of "SSTs".


In order to have adequate properties for pumped runs of "SSTs",
the notion of loop needs to be defined so as to take into 
account how the content of variables ``flows'' into 
other variables when performing an "update". 
\AP 
We define the ""skeleton"" of an "update" $\alpha: \var\rightarrow (\alp \uplus \var)^*$
as the "update" $\hat\alpha: \var \rightarrow \var^*$ obtained from $\alpha$
by removing all the letters of $\Sigma$ from the right-hand side.
\AP
Note that there are only finitely many "skeletons", and 
their composition forms a finite monoid, called the
""skeleton monoid""
(this notion is very similar to the flow monoid from \cite{equivalence-finite-valued},
 but does not rely on any normalization).

\AP
A ""loop"" of a "run" $\rho$ of an "SST" 
is any factor $L$ of $\rho$ 
that starts and ends in the same state and 
"induces@induced update" a ""skeleton-idempotent"" "update",
namely, an "update" $\alpha$ such that $\alpha$ and $\alpha\,\alpha$
have the same "skeleton". 
For example, the "update" 
$\alpha: \: X_1\mapsto a \, X_1 \, b \, X_2 \, c, \:  X_2\mapsto a$
is "skeleton-idempotent" and thus can be part of a "loop".
A "loop" in a "run" will be denoted by an interval $[i,j]$.
In this case, it is convenient to assume that the indices $i,j$ represent
``positions'' in-between the transitions, thus identifying occurrences of 
states; in this way, adjacent "loops" can be denoted by intervals of 
the form $[i_1,i_2]$, $[i_2,i_3]$, etc.
In particular, if the run consists of $n$ transitions, then 
the largest possible interval on it is $[0,n]$.
\AP
For technical reasons, we do allow ""empty@@loop"" "loops",
that is, "loops" of the form $[i,j]$, with $i=j$ and with 
the "induced update" being the identity function on $\var$.

\AP
The run obtained from $\rho$ by pumping $n$ times a "loop" $L$ 
is denoted $""\pump*<n>[L]{\rho} @pumping""$.
If we are given an $m$-tuple of \emph{pairwise disjoint} "loops"
$\bar{L}=(L_1,\dots,L_m)$ and an $m$-tuple of (positive) numbers 
$\bar{n}=(n_1,\dots,n_m)$, then we write
$\reintro[pumping]{\pump*<\bar n>[\bar L]{\rho}}$ 
for the "run" obtained by pumping simultaneously $n_i$ times
$L_i$, for each $1 \le i \le m$.

The next lemma is a Ramsey-type argument that,
based on the number of states of the "SST", the size of the 
"skeleton monoid", and a number $n$,
derives a minimum length for a "run" to witness $n+1$ points
and "loops" between pairs of any of these points. 
The reader can refer to \cite{JeckerSTACS21}
to get good estimates of the values of $E,H$.

\begin{lemma}\label{lem:ramsey}
Given an "SST", one can compute two numbers $E,H$ such that
for every run $\rho$, every $n\in\mathbb{N}$, and every set
$I \subseteq \{0,\dots,|\rho|\}$ of cardinality $E n^H+1$, 
there is a subset $I'\subseteq I$ of cardinality $n + 1$
such that for all $i < j \in I'$ the interval
$[i,j]$ is a "loop" of~$\rho$.
The values of $E,H$ are elementary in the size of the "SST".
\end{lemma}

\begin{proof}
  The article \cite{JeckerSTACS21} shows for any given monoid $M$ a bound
  $R_M(n)$ with the property that any sequence from $M^*$ of length larger than
  $R_M(n)$ contains $n$ consecutive infixes such that for some idempotent
  $e$ of $M$ (i.e., satisfying $ee =e$) each such infix multiplies out
  to $e$.
  In our case, the monoid $M$ is the product of the 
  monoid $((Q \times Q) \cup \{0\},\cdot)$ with $(p,q)\cdot (q,r)=(p,r)$
  (resp.~$(p,q)\cdot (q',r)=0$ if $q \not= q'$) and the "skeleton monoid" of the "SST".
  Thus, any infix of the run that multiplies out to an idempotent
  (after mapping each transition to the corresponding monoid element)
  corresponds to  a "loop" of the "SST".
  The upper bound $R_M(n)$ is  exponential only in the size of $M$ (cf.~Theorem 1 in \cite{JeckerSTACS21}).
\end{proof}

Below, we describe the effect on the "output" of "pumping" 
"loops" in a "run" of an "SST".
We start with the following simple combinatorial result:

\begin{lemma}\label{lem:loop-pumping}
Let $\alpha$ be a "skeleton-idempotent" "update".
For every variable $X$, there exist two words $u,v \in \Sigma^*$ 
such that, for all positive natural numbers $n \in \mathbb{N}_+$, 
$\alpha^n(X)=u^{n-1} \, \alpha(X) \, v^{n-1}$.
\end{lemma}

\begin{proof}
    Let $\hat{\alpha}$ be the "idempotent" "skeleton" of $\alpha$. 
    We first prove the following claim:
    
    \begin{claim*}
    For all $X\in\var$, if $\hat\alpha(X) \neq \varepsilon$,
    then $X$ occurs in $\alpha(X)$.
    \end{claim*}

\begin{proof}[Proof of the claim]
The proof is by induction on the number of variables $X$
such that $\hat\alpha(X) \neq \varepsilon$.
The base case holds vacuously.
As for the induction step, fix a variable $X_0$ such that
$\hat\alpha(X_0) \neq \varepsilon$.
We distinguish two cases:
\begin{itemize}
\item If $X_0$ occurs in $\alpha(X_0)$, then the claim clearly holds for $X_0$
      and it remains to prove it for all other variables $Y\in\var\setminus\{X_0\}$.
      Since $\alpha$ is "copyless",
      $X_0$ does not occur in $\alpha(Y)$, for all $Y\in\var\setminus\{X_0\}$.
      Therefore, the restricted "update" $\alpha|_{\var\setminus\{X_0\}}$ 
      has an "idempotent" "skeleton", and by the inductive hypothesis 
      it satisfies the claim.
      From this we immediately derive that the claim also holds for the 
      original "update" $\alpha$.
\item If $X_0$ does not occur in $\alpha(X_0)$, then 
      $\alpha(X_0)$ still contains at least one occurrence of another variable,
      since $\hat\alpha(X_0) \neq \varepsilon$.
      So, suppose that $\alpha(X_0) = h\, Y \, t$ 
      for some $h,t\in (\var\cup\Sigma)^*$ and $Y\in\var$. 
      Since $\alpha$ is "skeleton-idempotent", we have
      $(\alpha\alpha)(X_0) = \alpha(h) \, \alpha(Y) \, \alpha(t) = h'\, Y\, t'$
      for some $h',t'\in (\var\cup\Sigma^*)$. 
      Now, if $Y$ occurs in $\alpha(Y)$, then we can
      apply the inductive hypothesis on the restriction
      $\alpha|_{\var\setminus\{Y\}}$, as in the previous case,
      reaching a contradiction for $X_0$.
      Otherwise, if $Y$ does not occur in $\alpha(Y)$, we also reach
      a contradiction by arguing as follows.
      Since $Y$ occurs in $\alpha(h) \, \alpha(Y) \, \alpha(t)$
      but not in $\alpha(Y)$, then it occurs in $\alpha(h)\,\alpha(t)$. 
      Since $X_0$ is the unique variable such that $Y$ occurs in $\alpha(X_0)$ 
      (because $\alpha$ is "copyless"), necessarily $X_0$ occurs in $h\, t$ 
      and hence in $\alpha(X_0)$ as well, which contradicts the initial 
      assumption. 
      \qedhere
\end{itemize}
\end{proof}
   
      We conclude the proof of the lemma. Let $X\in \var$. 
      If $\alpha(X)$ does not contain any variable, then
      we immediately get the result, as
      $\alpha^n(X) = \alpha(X)$ for all $n\geq 1$. 
      Otherwise, if $\alpha(X)$ contains some variable, then
      we know that $\hat\alpha(X)\neq\varepsilon$, 
      so by the above claim, $X$ occurs in $\alpha(X)$. 
      Hence, $\alpha(X) = h \, X \,t$ 
      for some $h,t\in(\var\cup\Sigma)^*$. 
      We now prove that
      $\alpha(h)\, \alpha(t)\in\Sigma^*$. 
      Indeed, let $Y$ be any variable in $h\,t$ 
      (if there is no such variable then clearly 
       $\alpha(h)\, \alpha(t)$ contains no variable either). 
      Since $\alpha$ is "copyless", 
      $Y$ does not occur in $\alpha(Y)$, 
      hence by the above claim (contrapositive), 
      $\hat\alpha(Y)=\varepsilon$,
      hence $\alpha(Y)\in\Sigma^*$. 
      We then derive
      $(\alpha\alpha)(X) = \alpha(h) \, \alpha(X) \, \alpha(t)$,
      where $\alpha(h),\alpha(t)\in\Sigma^*$, and so we can take 
      $u = \alpha(h)$ and $v=\alpha(t)$. 
      To conclude, we have 
      $\alpha^2(X) = u \, \alpha(X) \,v$, so for $n>2$, 
      $\alpha^n(X) = \alpha^{n-2}(\alpha^2(X)) = 
       \alpha^{n-2}(u \, \alpha(X) \,v) = u \, \alpha^{n-1}(X) \,v$,
      as claimed. 
\end{proof}

It follows that "pumping" "loops" in a "run" corresponds 
to introducing repeated copies of factors in the "output". 
Similar results can be found
in~\cite{equivalence-finite-valued} for "SSTs" 
and in \cite{rozoy1986,GaetanPhD} for two-way transducers:

\begin{corollary}\label{cor:loop-pumping}
Let $\rho$ be an "accepting" "run" of an "SST" and let $\bar L = (L_1,\dots,L_m)$ 
be a tuple of pairwise disjoint "loops" in $\rho$.
Then, for some $r \le 2m|\var|$ there exist words $w_0,\dots,w_r,u_1,\dots,u_r$ 
and indices $1\le i_1,\dots,i_r\le m$, not necessarily distinct, such that 
for every tuple $\bar n = (n_1,\dots,n_m) \in \mathbb{N}^m_+$ of positive natural numbers, 
\[
  \out{\pump<\bar n>[\bar L]{\rho}}
  ~=~
  w_0 \, u_1^{n_{i_1}-1} \, w_1 \, \dots \, u_r^{n_{i_r}-1} \, w_r.
\]
\end{corollary}

\begin{proof}
This follows immediately from Lemma \ref{lem:loop-pumping}.
Note that the content of any variable $X$ just after
"pumping" a "loop" either appears as infix of the final output, 
or is erased by some later "update". 
In both cases, each "pumped" "loop" $L_i$ induces in the "output" 
$(n_i-1)$-folded repetitions of $2k$ (possibly empty) factors, 
where $k$ is the number of variables of the "SST".
Since the "loops" are pairwise disjoint, they contribute such factors
without any interference.
The final output $\out{\pump<\bar n>[\bar L]{\rho}}$ 
thus features repetitions of $r = 2km$ (possibly empty) factors.
\end{proof}

The rest of the section analyses properties of words
with repeated factors like the one in Corollary \ref{cor:loop-pumping}.

\begin{definition}\label{def:word-equation}
A ""word inequality"" with repetitions parametrized in 
$\formal{X}$ 
is a pair $e = (w, w')$ of terms of the form
\begin{align*}
w &~=~ s_0 ~ t_1^{\formal{x}_1} ~ s_1 ~ \dots ~ t_m^{\formal{x}_m} ~ s_m \\
w' &~=~ s'_0 ~ {t'_1}^{\formal{y}_1} ~ s'_1 ~ \dots ~
{t'_n}^{\formal{y}_n} ~ {s'_n}
\end{align*}
where $s_i,t_i, s'_j, t'_j \in \Sigma^*$ and 
$\formal{x}_i, \formal{y}_j \in \formal{X}$
for all $i,j$.
\AP
The set of ""solutions"" of $e = (w,w')$, denoted $\reintro[solutions]{\sol*{e}}$, 
consists of the mappings $f:\formal{X}\rightarrow\mathbb{N}$ 
such that $f(w) \neq f(w')$, where $f(w)$ is the word obtained 
from $w$ by substituting every formal parameter $\formal x\in\formal{X}$ 
by $f(\formal x)$, and similarly for $f(w')$.
\AP
A ""system of word inequalities"" 
is a non-empty finite set $E$ of "inequalities" as above,
and its set of "solutions" is given by
$\reintro[solutions]{\sol*{E}} = \bigcap_{e\in E} \sol{e}$.
\end{definition}

The next theorem states that if there exists a solution to a system of inequalities parameterized by a single 
variable $\formal{x}$, then the set of solutions is co-finite. 

\begin{theorem}[name={\cite[Theorem 4.3]{Saarela15}}]\label{thm:saarela}
Given a "word inequality" $e$ with repetitions parameterized by single 
variable $\formal{x}$,
$\sol{e}$ is either empty or co-finite; more precisely,
if the left (resp.~right) hand-side of $e$ contains $m$ (resp.~$n$)
repeating patterns (as in Definition \ref{def:word-equation}),
then either $\sol{e} = \emptyset$ or $|\mathbb{N} \setminus \sol{e}| \le m+n$.
\end{theorem}

Finally, we present two corollaries of the above theorem, 
that will be used later. 
The first corollary concerns "satisfiability" of a "system of inequalities".
\AP
Formally, we say that a "word inequality" $e$ 
(resp.~a "system of inequalities" $E$) 
is ""satisfiable"" if its set of "solutions" is non-empty.

\begin{corollary}\label{cor:saarela1}
Let $E$ be a finite "system of word inequalities".
If every "inequality" $e\in E$ is "satisfiable", then so is the "system" $E$.
\end{corollary}

\begin{figure}[t]
\centering
\begin{tikzpicture}[scale=1.5]
\tikzstyle{dot} = [draw, shape=circle, fill, minimum size=2mm, inner sep=0pt, outer sep=0pt]
\tikzstyle{circle} = [draw, thick, shape=circle, minimum size=2mm, inner sep=0pt, outer sep=0pt]
\tikzstyle{arrow} = [->, >=stealth, shorten >=1pt]
\path [gray, arrow] (-0.25,0) edge (7,0);
\path [gray, arrow] (0,-0.25) edge (0,6);
\path [gray, arrow] (7, 0) node [below = 1mm] {$\formal x$};
\path [gray, arrow] (0,6) node [left = 1mm] {$\formal y$};

\node [dot, red] (x') at (1,2) {};
\node [dot, blue] (x'') at (2,1) {};
\node [dot, red] (y') at (3,2) {};
\node [dot, red!50, scale=0.75] at (3.5,2) {};
\node [dot, red!50, scale=0.5] (y'bis) at (4,2) {};
\node [dot, red!50, scale=0.325] at (4.5,2) {};
\node [dot, blue] (y'') at (4,1) {};
\node [dot, blue!50, scale=0.75] at (4.5,1) {};
\node [dot, blue!50, scale=0.5] at (5,1) {};
\node [dot, blue!50, scale=0.325] at (5.5,1) {};
\node [dot, red] (z') at (4,4.5) {};
\node [dot, red!50, scale=0.75] at (4,5) {};
\node [dot, red!50, scale=0.5] at (4,5.5) {};
\node [dot, red!50, scale=0.325] at (4,6) {};
\node [dot, blue] (z'') at (4,3.5) {};
\node [dot, blue!50, scale=0.75] at (4,4) {};
\node [dot, blue!50, scale=0.5] at (4,4.5) {};
\node [dot, blue!50, scale=0.325] at (4,5) {};

\node [circle, black, scale=1.25] at (x') {};
\node [circle, black, scale=1.25] at (x'') {};
\node [circle, black, scale=1.25] at (y'bis) {};
\node [circle, black, scale=1.25] at (y'') {};
\node [circle, black, scale=1.25] (z) at (4,4.5) {};

\node [above=3mm] at (x') {$g'_{\emptyset}$};
\node [below=3mm] at (x'') {$g''_{\emptyset}$};
\node [above=3mm] at (y'bis) {$\phantom{X}g'_{\{\formal x\}}$};
\node [below=3mm] at (y'') {$g''_{\{\formal x\}}$};
\node [right=3mm] at (z) {$g'_{\{\formal x,\formal y\}} = g''_{\{\formal x,\formal y\}}$};

\path [red, arrow] (x') edge [bend left] node [above=0.2mm, red] {\small Th.~\ref{thm:saarela}} (y');
\path [blue, arrow] (x'') edge [bend right] node [below=0.8mm, blue] {\small Th.~\ref{thm:saarela}} (y'');
\path [red, arrow] (y'bis) edge [bend left=45] node [above=0.6mm, red, rotate=90] {\small Th.~\ref{thm:saarela}} (z');
\path [blue, arrow] (y'') edge [bend right=60] node [above=0.6mm, blue, rotate=-90] {\small Th.~\ref{thm:saarela}} (z'');
\end{tikzpicture}
\caption{Illustration of an argument for the proof of Corollary \ref{cor:saarela1}}\label{fig:saarela1}
\end{figure}
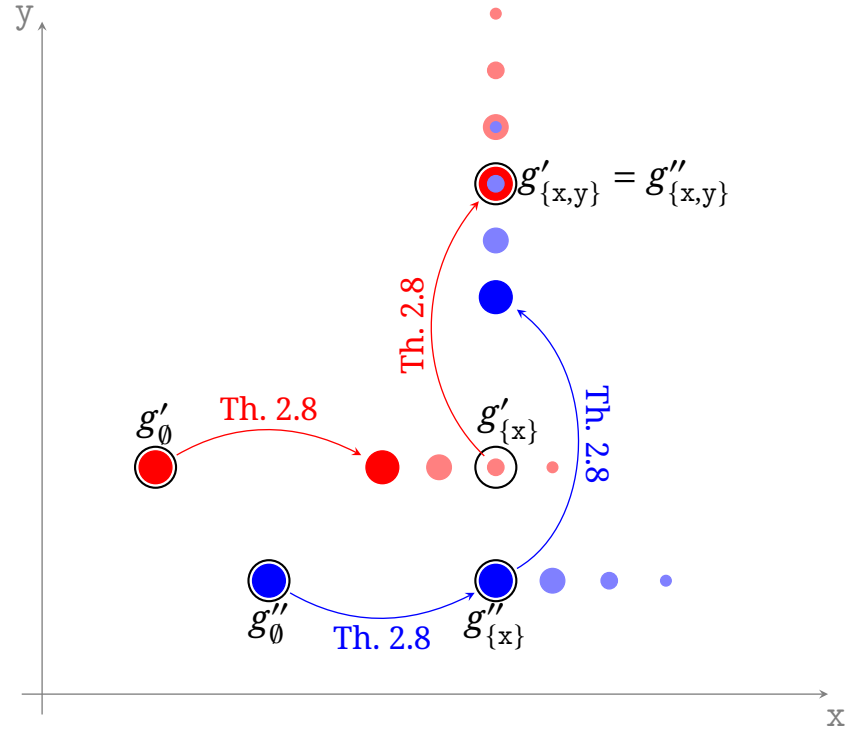

\begin{proof}
All inequalities considered hereafter have 
parameters in $\formal X=\{\formal{x}_1,\dots,\formal{x}_k\}$.
We are going to compare functions from $\formal X$ to $\mathbb{N}$
based on suitable partial orders, each parametrized by a variable. 
\AP
Formally, given two functions $f,g:\formal X\rightarrow\mathbb{N}$
and a variable $\formal x\in\formal X$, we write 
$f \mathrel{""\plex*{\formal x}@variable order""} g$ 
iff $f(\formal x) \le g(\formal x)$ and $f(\formal y) = g(\formal y)$ 
for all $\formal y\in\formal X\setminus\{\formal x\}$.
We prove the following two properties (the first property is equivalent to 
the claim of the lemma):
\begin{align}
  &\bigwedge_{e\in E}
    \exists f \in \sol{e}
    ~~\rightarrow~~
   \exists g \in \sol{E}
  \label{eq:exists-solution}
  \\
  &\forall f \in \sol{E} ~~
  \forall \formal x\in\formal X ~~
  \exists g \pgex{\formal x} f ~~
  \forall h \pgex{\formal x} g ~~:~~
  h \in \sol{E} 
  \label{eq:forall-solutions}
\end{align}
The proof goes by  double induction on the cardinality of $E$ 
and the number $k$ of parameters.

The base case is when $E$ has cardinality $1$.
In this case Property \eqref{eq:exists-solution} holds trivially.
We see how Property \eqref{eq:forall-solutions} follows from Theorem~\ref{thm:saarela}.
Let $E=\{e\}$, $f\in\sol{e}$, and fix an arbitrary variable $\formal x\in\formal X$.
We construct the "inequality" $e'$ with $\formal x$ as single formal parameter, by 
instantiating in $e$ every other parameter $\formal y\in\formal X\setminus\{\formal x\}$
with the value $f(\formal y)$. 
By construction, $f$ restricted to $\{\formal x\}$ is a "solution" of $e'$, and thus, 
by Theorem \ref{thm:saarela}, $\sol{e'}$ is co-finite. 
This means that there is a number $x_0\in\mathbb{N}$ such that $x_0 \ge f(\formal x)$ 
and, for all $x_1 \ge x_0$, the mapping $\formal x\mapsto x_1$ is a "solution" to $e'$ as well.
This property can be transferred to the original "inequality" $e$, as follows.
We define $g = f[\formal x / x_1]$ as the function obtained from $f$ by replacing
the image of $\formal x$ with $x_1$. Note that $g \ge_{\formal x} f$ and,
for all $h \pgex{\formal x} g$, $h\in\sol{e}$. This proves Property \eqref{eq:forall-solutions}.

As for the inductive step, suppose that $E$ is a "system" with at least two "inequalities",
and divide $E$ into two sub-systems, $E'$ and $E''$, with cardinalities strictly 
smaller than that of $E$.

Let us first prove Property \eqref{eq:exists-solution} for $E$.
Suppose that every "inequality" in $E$ is "satisfiable".
By the inductive hypothesis, $E'$ and $E''$ are also "satisfiable";
in particular, there exist "solutions" $g'$ and $g''$ of $E'$ and $E''$, respectively.
We proceed with a second induction to prove that, for larger and larger
sets of variables $\formal Y\subseteq\formal X$, there are "solutions"
of $E'$ and $E''$ that agree on all the variables from $\formal Y$, namely:
\begin{align*}
  \exists g'_{\formal Y}\in\sol{E'} ~~
  \exists g''_{\formal Y}\in\sol{E''} ~~
  \forall \formal y\in \formal Y ~~
  g'_{\formal Y}(\formal y) = g''_{\formal Y}(\formal y) \ .
  \tag{$\star$}
\end{align*}
Of course, for $\formal Y=\formal X$, the above property will imply 
the existence of a "solution" of $E$.
The reader can refer to Figure \ref{fig:saarela1} as an illustration 
of the arguments that follow (axes correspond to variables, and red 
and blue dots represent "solutions" of the "systems" $E'$ and $E''$,
respectively).

For $\formal Y=\emptyset$, the claim ($\star$) is trivial, 
since we can simply let $g'_\emptyset = g'$ and $g''_\emptyset = g''$.
For the inductive step, suppose that ($\star$) holds for $\formal Y$
and let us prove it also holds for $\formal Y' = \formal Y \uplus \{\formal x\}$. 
By inductive hypothesis, $E'$ and $E''$ satisfy Property \eqref{eq:forall-solutions}.
In particular, by instantiating $f$ with $g'_{\formal Y}$ (resp.~$g''_{\formal Y}$)
in Property \eqref{eq:forall-solutions}, 
we obtain the existence of $g' \pgex{\formal x} g'_{\formal Y}$ such that, 
for all $h' \pgex{\formal x} g'$, $h'\in\sol{E'}$
(resp.~$g'' \pgex{\formal x} g''_{\formal Y}$ such that, 
for all $h'' \pgex{\formal x} g''$, $h''\in\sol{E''}$). 
Note that the functions $g'_{\formal Y}$, $g''_{\formal Y}$, $g'$, $g''$ 
all agree on the variables from $\formal Y$.
Moreover, without loss of generality, we can assume that $g'$ and $g''$ also agree
on the variable $\formal x$: indeed, if this were not the case, we could simply 
replace the $\formal x$-images of $g'$ and $g''$ with $\max\{g'(\formal x),g''(\formal x)\}$,
without affecting the previous properties.
Property ($\star$) now follows from letting $g'_{\formal Y'} = g'$ and $g''_{\formal Y'} = g''$.
This concludes the proof of the inductive step for Property \eqref{eq:exists-solution}.

Let us now prove the inductive step for Property \eqref{eq:forall-solutions}.
Let $f$ be a "solution" of $E$ and let $\formal x\in\formal X$.
Since both $E'$ and $E''$ satisfy Property \eqref{eq:forall-solutions} 
and since $f\in\sol{E'} \cap \sol{E''}$, there are 
$g',g'' \pgex{\formal x} f$ such that, 
for all $h' \pgex{\formal x} g'$ and $h'' \pgex{\formal x} g''$,
$h'\in\sol{E'}$ and $h''\in\sol{E''}$.
Since $g',g'' \pgex{\formal x} f$, 
$g'$ and $g''$ agree on all variables, except possibly $\formal x$.
Without loss of generality, we can also assume that $g'$ and $g''$ 
agree on $\formal x$: as before, if this were not the case, 
we could replace the $\formal x$-images of $g'$ and $g''$ with 
$\max\{g'(\formal x),g''(\formal x)\}$, while preserving the previous properties. 
Now that we have $g'=g''$, we can use this function to witness
Property \eqref{eq:forall-solutions}, since, for all $h \pgex{\formal x} g'$ ($ = g''$),
we have $h\in \sol{E'} \cap \sol{E''} = \sol{E}$.
\end{proof}

The second corollary is related to the existence of large sets of 
solutions for a "satisfiable" "word inequality" that avoid any 
correlation between variables. 
To formalize the statement, it is convenient to fix a total order 
on the variables of the "inequality", 
say~$\formal x_1,\dots,\formal x_k$,
and then identify every function $f:\formal X\rightarrow\mathbb{N}$ 
with the $k$-tuple of values $\bar x = (x_1,\dots,x_k)$, where 
$x_i = f(\formal x_i)$ for all $i=1,\dots,k$.
According to this correspondence, the corollary states the existence 
of sets of solutions that look like Cartesian products of finite
intervals of values, each with arbitrarily large cardinality.
The statement of the corollary is in fact slightly more complicated 
than this, as it discloses dependencies between the intervals.
We also observe that the order in which we list the variables 
is arbitrary, but different orders will induce different 
dependencies between intervals.

\begin{corollary}\label{cor:saarela2}
Let $e$ be a "word inequality" with repetitions parametrized in 
$\formal X = \{\formal x_1,\dots,\formal x_k\}$.
If $e$ is "satisfiable", then 
\begin{align*}
  & \exists \ell_1 ~ \forall h_1 ~ \dots ~ \exists \ell_k ~ \forall h_k \\
  & \quad
  \underbrace{[\ell_1,h_1]}_{\text{values for }\formal x_1}
  \times \, \dots \, \times
  \underbrace{[\ell_k,h_k]}_{\text{values for }\formal x_k} 
  \subseteq~ \sol{e}.
\end{align*}
\end{corollary}

\begin{figure}[t]
\centering
\begin{tikzpicture}[scale=1.5]
\tikzstyle{dot} = [draw, shape=circle, fill, minimum size=2mm, inner sep=0pt, outer sep=0pt]
\tikzstyle{circle} = [draw, thick, shape=circle, minimum size=2mm, inner sep=0pt, outer sep=0pt]
\tikzstyle{arrow} = [->, >=stealth, shorten >=1pt]
\path [gray, arrow] (-0.25,-0.5) edge (7,-0.5);
\path [gray, arrow] (0,-0.75) edge (0,6.5);
\path [gray, arrow] (7, 0) node [below = 1mm] {$\formal x_1$};
\path [gray, arrow] (0,6.5) node [left = 1mm] {$\formal x_2$};

\node [dot, black] (x) at (0.5,1) {};

\node [below=3mm] at (x) {$\phantom{XXX}\{x_1\}\times\{x_2\}$};

\node [dot, black!100, scale=1.0] (y1) at (3.5,1) {};
\node [dot, black!100, scale=1.0] (y2) at (4,1) {};
\node [dot, black!100, scale=1.0] (y3) at (4.5,1) {};
\node [dot, black!100, scale=1.0] (y4) at (5,1) {};
\node [dot, black!50, scale=0.75] (y5) at (5.5,1) {};
\node [dot, black!50, scale=0.5] (y6) at (6,1) {};
\node [dot, black!50, scale=0.325] (y7) at (6.5,1) {};

\path [gray, dashed] (y1) edge (3.5,-0.75);
\path [gray, dashed] (y4) edge (5,-0.75);
\path [black] (3.5,-0.75) node [below = 1mm] {$\exists\ell_1$};
\path [black] (5,-0.75) node [below = 1mm] {$\forall h_1$};

\node [below right=4mm] at (y4) {$[\ell_1,h_1]\times\{x_2\}$};

\node [dot, black!50, scale=1.0] (z11) at (3.5,2.5) {};
\node [dot, black!50, scale=1.0] (z12) at (3.5,3) {};
\node [dot, black!50, scale=1.0] (z13) at (3.5,3.5) {};
\node [dot, black!100, scale=1.0] (z14) at (3.5,4) {};
\node [dot, black!100, scale=1.0] (z15) at (3.5,4.5) {};
\node [dot, black!100, scale=1.0] (z16) at (3.5,5) {};
\node [dot, black!50, scale=0.75] (z17) at (3.5,5.5) {};
\node [dot, black!50, scale=0.5] (z18) at (3.5,6) {};
\node [dot, black!50, scale=0.325] (z19) at (3.5,6.5) {};

\node [dot, black!50, scale=1.0] (z21) at (4,3) {};
\node [dot, black!50, scale=1.0] (z22) at (4,3.5) {};
\node [dot, black!100, scale=1.0] (z23) at (4,4) {};
\node [dot, black!100, scale=1.0] (z24) at (4,4.5) {};
\node [dot, black!100, scale=1.0] (z25) at (4,5) {};
\node [dot, black!50, scale=0.75] (z26) at (4,5.5) {};
\node [dot, black!50, scale=0.5] (z27) at (4,6) {};
\node [dot, black!50, scale=0.325] (z28) at (4,6.5) {};

\node [dot, black!50, scale=1.0] (z31) at (4.5,3.5) {};
\node [dot, black!100, scale=1.0] (z32) at (4.5,4) {};
\node [dot, black!100, scale=1.0] (z33) at (4.5,4.5) {};
\node [dot, black!100, scale=1.0] (z34) at (4.5,5) {};
\node [dot, black!50, scale=0.75] (z35) at (4.5,5.5) {};
\node [dot, black!50, scale=0.5] (z36) at (4.5,6) {};
\node [dot, black!50, scale=0.325] (z37) at (4.5,6.5) {};

\node [dot, black!100, scale=1.0] (z41) at (5,4) {};
\node [dot, black!100, scale=1.0] (z42) at (5,4.5) {};
\node [dot, black!100, scale=1.0] (z43) at (5,5) {};
\node [dot, black!50, scale=0.75] (z44) at (5,5.5) {};
\node [dot, black!50, scale=0.5] (z45) at (5,6) {};
\node [dot, black!50, scale=0.325] (z46) at (5,6.5) {};

\path [gray, dashed] (z14) edge (-0.25,4);
\path [gray, dashed] (z16) edge (-0.25,5);
\path [black] (-0.25,4) node [left = 1mm] {$\exists \ell_2$};
\path [black] (-0.25,5) node [left = 1mm] {$\forall h_2$};

\node [right=3mm] at (z43) {$[\ell_1,h_1]\times[\ell_2,h_2]$};

\path [black, arrow] (x) edge [bend left] node [above=0.2mm] {\small Th.~\ref{thm:saarela}} (y1);
\path [black, arrow] (y1) edge [bend right] node [above=0.2mm, rotate=-90] {\small Th.~\ref{thm:saarela}} (z11);
\path [black, arrow] (y2) edge [bend right] node [above=0.2mm, rotate=-90] {\small Th.~\ref{thm:saarela}} (z21);
\path [black, arrow] (y3) edge [bend right] node [above=0.2mm, rotate=-90] {\small Th.~\ref{thm:saarela}} (z31);
\path [black, arrow] (y4) edge [bend right] node [above=0.2mm, rotate=-90] {\small Th.~\ref{thm:saarela}} (z41);
\end{tikzpicture}
\caption{Illustration of an argument for the proof of Corollary \ref{cor:saarela2}}\label{fig:saarela2}
\end{figure}
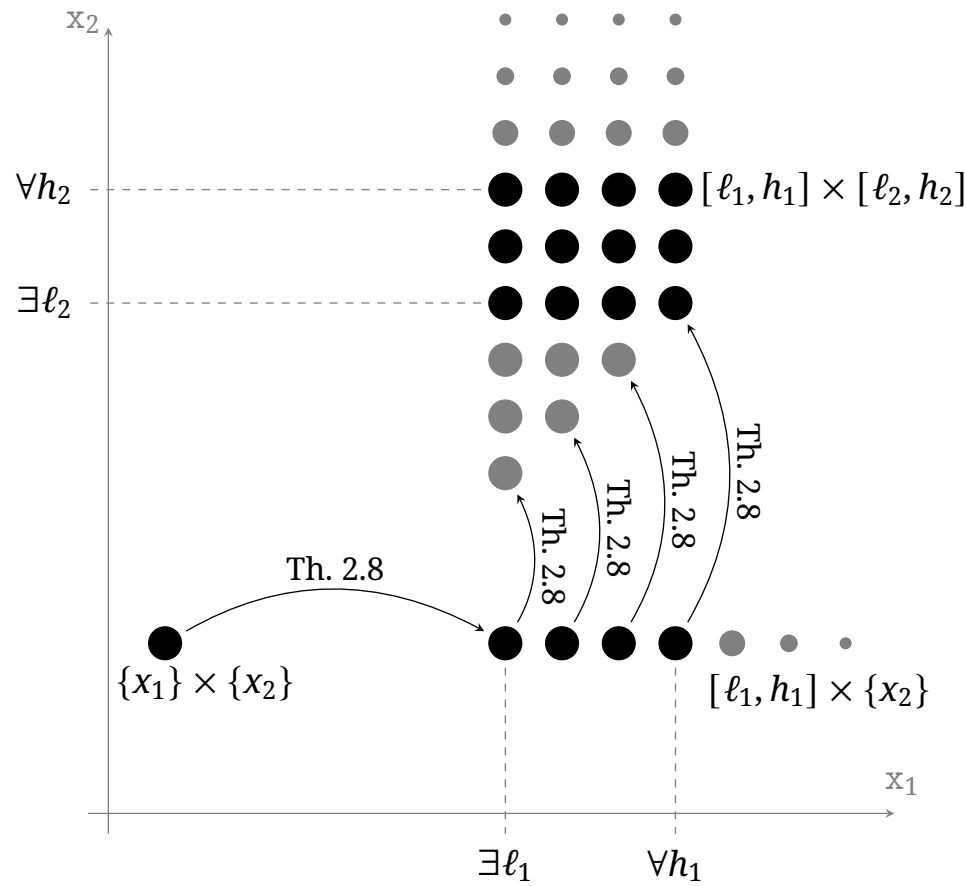

\begin{proof}
Let $e$ be a "satisfiable" "word inequality" parametrized in $\formal X$ and
let $\bar x = (x_1,\dots,x_k)$ be any "solution" of $e$.
We will prove the following claim by induction on $i=0,\dots,k$:
\begin{align*}
 & \exists \ell_1 ~ \forall h_1 ~ \dots ~ \exists ~ \ell_i ~ \forall h_i\\
 & \quad
 [\ell_1,h_1]\times\cdots\times[\ell_i,h_i]\times\{x_{i+1}\}\times\dots\times\{x_k\} \subseteq \sol{e}
 \tag{$\star$}
\end{align*}
Note that for $i=k$ the above claim coincides with the statement of the corollary.
The reader can also refer to Figure \ref{fig:saarela2} as an illustration 
of the arguments that follow (axes correspond to variables $\formal x_1$ 
and $\formal x_2$, dots represent generic "solutions" of $e$, 
clusters of black dots represent solutions in the form of Cartesian products,
like those that appear in ($\star$).

For the base case $i=0$, the claim ($\star$) is vacuously true, as $\bar x=(x_1,\dots,x_k)$
is a "solution" of $e$.

For the inductive step, we need to show that 
if ($\star$) holds for $i<k$, then it also holds for $i+1$.
It is in fact sufficient to prove that, for $i<k$,
\[
  [\ell_1,h_1]\times\cdots\times[\ell_i,h_i]\times
  \{x_{i+1}\}\times\dots\times\{x_k\} 
  \subseteq \sol{e}
\]
implies
\begin{align*}
  & \exists \ell_{i+1} ~ \forall h_{i+1} \\
  & [\ell_1,h_1]\times\cdots\times[\ell_{i+1},h_{i+1}]\times
    \{x_{i+2}\}\times\dots\times\{x_k\} 
    \subseteq \sol{e}.
\end{align*}
For brevity, we let 
$S = [\ell_1,h_1]\times\dots\times[\ell_i,h_i]\times    
     \{x_{i+1}\}\times\dots\times\{x_k\}$,
and we assume that $S\subseteq \sol{e}$.
For every tuple $\bar s\in S$, we consider the "word inequality" 
$e_{\bar s}$ over a single
variable $\formal x_{i+1}$ that is obtained from $e$ by instantiating every other variable
$\formal x_j$ ($j\neq i$) with $\bar s[j]$.
Since $S\subseteq\sol{e}$, we know that $e_{\bar y}$ is "satisfiable", and hence
by Theorem \ref{thm:saarela}, $e_{\bar s}$ has co-finitely many "solutions".
This means that there is $\ell_{\bar s}$ such that, for all $x'\ge \ell_{\bar s}$, 
$x'$ is also a "solution" of $e_{\bar s}$.
This property can be transferred to our original "inequality" $e$:

\begin{claim*}\label{claim:saarela3}
For every $\bar s\in S$, there is $\ell_{\bar s}$ such that,
for every $x'\ge \ell_{\bar s}$, the tuple $\bar s[i+1 \mapsto x']$, 
obtained from $\bar s$ by replacing the $(i+1)$-th value with $x'$,
is a "solution" of $e$.
\end{claim*}

Now, 
the existentially quantified value $\ell_{i+1}$ 
can be set to
the maximum of the $\ell_{\bar s}$'s, for all $\bar s\in S$
(for this definition to make sense, it is crucial that the set $S$ is finite).
In this way, thanks to the previous claim, the containment 
$[\ell_1,h_1]\times\cdots\times[\ell_{i+1},h_{i+1}]\times\{x_{i+2}\}\times\dots\times\{x_k\} \subseteq \sol{e}$
holds for all choices of the universally quantified value $h_{i+1}$.
This proves the inductive step for $(\star)$ from $i$ to $i+1$.
\end{proof}

\subsection{Delay between accepting runs}\label{subsec:delay}

We briefly recall the definitions from \cite{SSTdelay}, that introduces
a measure of similarity (called "delay")
between "accepting" "runs" of an "SST" that have the same "input" and the same "output".

We first give some intuition, followed by definitions and an example.
Naturally, the difference between the amount of output symbols produced during a run should be an indicator of (dis)similarity. 
However, as "SSTs" do not necessarily build their output 
from left to right, one must also take into account 
the position where an output symbol is placed.
For example, compare two "runs" $\rho$ and $\rho'$ 
on the same "input" that produce the same "output" 
$aaabbb$. 
After consuming a prefix of the input, 
$\rho$ may have produced 
$aaa\_\ \_\ \_$ 
and $\rho'$ may have produced 
$\_\ \_\ \_bbb$. 
The amount of produced output symbols is the same,
but the "runs" are delayed because $\rho$ 
built the "output" from the left, 
whereas $\rho'$ did it from the right.
This idea of "delay" comes with an important caveat.
As another example, consider two runs $\rho$ and $\rho'$ 
on the same "input" that produce the same "output" 
$aaaaaa$,
and assume that, after consuming the same prefix of the "input",
$\rho$ and $\rho'$ produced 
$aaa\_\ \_\ \_$ and $\_\ \_\ \_aaa$, respectively. 
Note that the "output" 
$aaaaaa$ is a periodic word. 
Hence, it does not matter if 
$aaa$ is appended or prepended to a word with period $a$. 
In general, one copes with this phenomenon 
by dividing the "output" into periodic parts,
where all periods are bounded by a well-chosen parameter $C$.
So, intuitively, the "delay" measures 
the difference between the numbers of output symbols that have 
been produced by the two "runs", up to the end of each of periodic factor. The number of produced output symbols is formally captured 
by a "weight" function, defined below, and the "delay" aggregates
the "weight" differences. 

\AP
For an
"accepting" "run" $\rho$, a position $t$ of $\rho$, 
and a position $j$ in the output $\out{\rho}$, we denote by
$""\weight*[t][j]{\rho}@weight""$ %
the number of output positions $j'\le j$ that are produced by 
the prefix of $\rho$ up to position $t$.
We use the above notation when $j$ witnesses
a change in a repeating pattern of the output. These changes in repeating patterns are called cuts, as formalized below.

Let $w$ be any non-empty word (e.g.~the output of $\rho$ or a factor of it).
\AP
The \reintro{primitive root} of $w$, denoted $""\rt*{w} @root""$, 
is the shortest word $r$ such that $w \in \{r\}^*$.
\AP
For a fixed integer $C>0$ we define a factorization 
$w[1,j_1],w[j_1+1,j_2],\dots,w[j_n + 1,j_{n+1}]$
of $w$ in which every $j_i$ is chosen as the rightmost 
position for which $w[j_{i-1}+1, j_i]$ has "primitive root" 
of length not exceeding $C$.
These positions $j_1,\dots,j_n$ are called ""$C$-cuts"".
More precisely:
\begin{itemize}
\item the \reintro[$C$-cuts]{first $C$-cut} of $w$ is the largest position $j\le |w|$, 
      such that $|\rt{w[1,j]}| \le C$;
\item if $j$ is the "$i$-th $C$-cut" of $w$,
      then the  \reintro[$C$-cuts]{$(i+1)$-th $C$-cut} of $w$ 
      is the largest position $j'>j$ such that 
      $|\rt{w[j+1,j']}| \le C$.
\end{itemize}
\AP
We denote by $""\cuts*<C>{w}@set of cuts""$ 
the set of all "$C$-cuts" of $w$.

\AP
We are now ready to define the notion of "delay".
Consider two "accepting" "runs" $\rho,\rho'$ of an "SST" 
with the \emph{same "input"} $u=\inp{\rho}=\inp{\rho'}$ and the 
\emph{same "output"} $w=\out{\rho}=\out{\rho'}$,
and define:
\[
  ""\delay*<C>{\rho,\rho'}@delay"" ~= 
  \max\limits_{\substack{t\le |u|, \\ j\,\in\,\cuts{w}}} ~ 
  \big| \weight[t][j]{\rho} - \weight[t][j]{\rho'} \big|
\]
In other words, the delay between two such runs $\rho$ and $\rho'$ measures over all input positions the maximal difference between the amount of generated output, but only up to "$C$-cuts"  of the "output".
Note that the "delay" is only defined for "accepting" "runs" with 
same "input" and "output". 
So whenever we write $\delay<C>{\rho,\rho'}$, we implicitly mean 
that $\rho,\rho'$ have \emph{same "input" and same "output"}.

\begin{example}\label{ex:delay}
Let $w = abcccbb$ be the "output" of runs $\rho, \rho'$ 
on the same "input" of length $2$. 
Assume $\rho$ produces $abc\_\ \_bb$ and then $abcccbb$, 
whereas $\rho'$ produces $\_\ \_\ \_c\_bb$ and then $abcccbb$.
For $C = 2$, we obtain $\cuts<2>{w} = \{2,5,7\}$, i.e., $w$ is divided into $ab|ccc|bb$.
To compute the $\delay<2>{\rho,\rho'}$, we need to calculate "weights"
at "cuts". For $t=0$, $\weight[0][j]{\rho} =
\weight[0][j]{\rho'} = 0$ for all $j\in \cuts<2>{w}$ because nothing
has been produced. For $t=2$, $\weight[2][j]{\rho} =
\weight[2][j]{\rho'} = j$ for all $j\in \cuts<2>{w}$ because the whole output
has been produced. Only the case $t=1$ has an impact on the delay. 
We have $\weight[1][2]{\rho} = 2$, $\weight[1][5]{\rho} = 3$, and $\weight[1][7]{\rho} = 5$.
Also, we have $\weight[1][2]{\rho'} = 0$, $\weight[1][5]{\rho} = 1$, and $\weight[1][7]{\rho} = 3$.
Hence, we obtain $\delay<2>{\rho,\rho'} = 2$. 
\end{example}

We recall below some crucial results obtained in \cite{SSTdelay}.
A first result shows that the relation of pairs of "runs" 
having bounded "delay" (for a fixed bound) is "automatic"
--- for this to make sense, we view a "run" of an "SST"
as a finite word, with letters representing transitions, 
and we recall that a relation is "automatic" 
if its "convolution" language is regular.

\begin{lemma}[{\cite[Theorem 5]{SSTdelay}}]\label{lem:regular-delay}
Given an "SST" and some numbers $C,D$, the relation consisting of
pairs of "accepting" "runs" $(\rho,\rho')$ such that $\delay<C>{\rho,\rho'} \le D$ 
is "automatic".
\end{lemma}

\begin{proof}
The statement in \cite[Theorem 5]{SSTdelay} is not for "runs" of "SSTs",
but for sequences of "updates". One can easily build an automaton 
that checks if two sequences of transitions, encoded by their "convolution",
form "accepting" "runs" $\rho,\rho'$ of the given "SST" on the same "input". 
The remaining condition $\delay<C>{\rho,\rho'} \le D$ only depends on 
the underlying sequences of "updates" determined by $\rho$ and $\rho'$,
and can be checked using \cite[Theorem 5]{SSTdelay}.
\end{proof}

A second result shows that given two "runs" with large
"delay", one can find a set of positions on the "input" 
(the cardinality of which depends on how large the "delay" is)
in such a way that any interval starting just before any of 
these positions and ending just after any other of these positions 
is a "loop" on both "runs" such that, when "pumped", produces 
different "outputs". 
Roughly, the reason for obtaining different "outputs" is that pumping creates a misalignment between "$C$-cuts" that were properly aligned before pumping, and different periods cannot overlap.
By this last result, large "delay" intuitively means 
``potentially different "outputs"''.

\begin{lemma}[name={\cite[Lemma 6]{SSTdelay}}]\label{lem:delay-vs-output}%
Given an "SST", one can compute%
\footnote{%
We remark that the notation and the actual bounds here 
differ from the original presentation of \cite{SSTdelay},
mainly due to the fact that here we manipulate runs
with explicit states and loops with idempotent skeletons.
In particular, the parameters $C,D,m$ here correspond respectively to 
the values $k E^2, \ell E^4, C E^2$ with $k,\ell,C$ as in \cite[Lemma 6]{SSTdelay}, and
 $E$ as in our Lemma \ref{lem:ramsey}.
}
some numbers $C,D$ such that, 
for all $m\ge 1$ and all runs $\rho,\rho'$:
if $\delay<C m>{\rho,\rho'} > D m^2$,
then there exist
$m$ positions $0\le \ell_1 < \dots < \ell_m \le |\rho|$ such that,
for every $1 \le i < j \le m$, the interval $L_{i,j}=[\ell_i,\ell_j]$
is a "loop" on both $\rho$ and $\rho'$ and satisfies
\[
  \out{\pump<2>[L_{i,j}]{\rho}} 
  \:\neq\: 
  \out{\pump<2>[L_{i,j}]{\rho'}} .
\]
\end{lemma}

To reason about "finite valuedness" we will need to consider 
\emph{several} "accepting" "runs" on the same "input", with pairwise 
large "delays".
By Lemma \ref{lem:delay-vs-output}, 
every two such "runs" can be pumped so as to witness different "outputs". 
The crux however is to show that these runs can be pumped \emph{simultaneously} 
so as to get pairwise different "outputs".
This is indeed possible thanks to: 

\begin{lemma}\label{lem:saarela-delay}\label{cor:saarela-delay}
Let $C,D$ be computed as in Lemma \ref{lem:delay-vs-output},
and $k$ be an arbitrary number. Then one can compute a number $m$
such that, for all runs $\rho_0,\dots,\rho_k$ on the same "input" and with 
$\displaystyle\bigwedge\nolimits_{0\le i<j\le k} 
 \big(\out{\rho_i} \neq \out{\rho_j} ~\vee~
      \delay<C m>{\rho_i,\rho_j} > D m^2\big)$,
there is a tuple 
$\bar{L} = (L_{i,j})_{0\le i<j\le k}$ of disjoint intervals
that are "loops" on all "runs" $\rho_0,\dots,\rho_k$,
and there is a tuple $\bar n = (n_{i,j})_{0\le i<j\le k}$
of positive numbers such that
\[
  \text{for all } 0\le i<j\le k\,, \quad 
  \out{\pump<\bar n>[\bar L]{\rho_i}} 
  \:\neq\:
  \out{\pump<\bar n>[\bar L]{\rho_j}} \, .
\]
\end{lemma} 
\begin{proof}
We first define $m$. 
Let $E,H$ be as in Lemma~\ref{lem:ramsey}, and set $m:=m_0+1$ for 
the sequence $m_0, \ldots, m_{k-1}$ defined inductively by 
$m_{k-1} := k(k+1)$, and 
$m_h := E m_{h+1}^H$. 

We show how to "pump" the "runs" in such a way that all pairs of indices $i<j$
witnessing $\delay<C m>{\rho_i,\rho_j} > D m^2$ before "pumping",
will witness different "outputs" after "pumping".
Consider one such pair $(i,j)$, with $i<j$, such that 
$\delay<C m>{\rho_i,\rho_j} > D m^2$, so in particular,
$\out{\rho_i} = \out{\rho_j}$
(if there is no such pair, then all "runs" have pairwise different
"outputs", and so we are already done).
We apply Lemma \ref{lem:delay-vs-output}
and obtain a set $I_{i,j,0}$ of $m = m_0+1$ positions such that each interval $L=[\ell,\ell']$ with $\ell,\ell' \in I_{i,j,0}$ is a loop on both $\rho_i$ and $\rho_j$, and:
\begin{equation}\label{eq:divergence1}
  \out{\pump<2>[L]{\rho_i}} \neq \out{\pump<2>[L]{\rho_j}}.
\end{equation}
Then, by repeatedly using Lemma~\ref{lem:ramsey}, 
we derive the existence of sets $I_{i,j,k-1} \subseteq  \cdots \subseteq I_{i,j,1} \subseteq I_{i,j,0}$ with $|I_{i,j,h}| = m_h+1$ such that each interval $L=[\ell,\ell']$ with $\ell,\ell' \in I_{i,j,h}$ is a loop on $\rho_i,\rho_j$, and $h$ further runs from $\rho_0,\ldots,\rho_k$ (our definition of $m$ from the 
beginning of the proof is tailored to this repeated application
of Lemma~\ref{lem:ramsey}, because $|I_{i,j,h}| = m_h+1 = E m_{h+1}^H +1$). In particular, all intervals with endpoints in $I_{i,j,k-1}$ are loops on all the $\rho_0,\ldots,\rho_k$.

In this way, for each pair $i<j$ such that $\rho_i$ and $\rho_j$ 
have large "delay", we obtain $k(k+1)$ adjacent intervals that are loops on all runs and that satisfy the pumping property (\ref{eq:divergence1}) from above.

As there are at most $k (k+1)$ pairs of runs, 
we can now choose from the sets of intervals that we have prepared 
one interval $L_{i,j}$ for each pair $i<j$ with $\delay<C m>{\rho_i,\rho_j} > D m^2$,
in such a way that all the chosen intervals are pairwise disjoint 
(for example, we could do so by always picking among the remaining intervals 
the one with the left-most right border, and then removing all intervals 
that intersect this one).
The selected intervals $L_{i,j}$ thus have the following properties:
\begin{enumerate}
\item $L_{i,j}$ is a "loop" on all "runs" $\rho_0,\dots,\rho_k$,
\item $L_{i,j}$ is disjoint from every other interval $L_{i',j'}$,
\item $\out{\pump<2>[L_{i,j}]{\rho_i}} \neq \out{\pump<2>[L_{i,j}]{\rho_j}}$.
\end{enumerate}
If a pair $i<j$ of "runs" is such that 
$\out{\rho_i} \neq \out{\rho_j}$, 
then we set $L_{i,j}$ as an empty "loop". 

Now, let $\bar{L} = (L_{i,j})_{0 \le i < j \le k}$ be the tuple of chosen intervals, and consider the following system of word "inequalities" with formal parameters $(\formal{x}_{i,j})_{0\le i < j \le k} =: \bar{\formal{x}}$:
\[
  \text{for all } 0\le i<j\le k\,, \quad 
  \out{\pump<\bar{\formal{x}}>[\bar L]{\rho_i}} 
  \:\neq\:
  \out{\pump<\bar{\formal{x}}>[\bar L]{\rho_j}} \, .
\]
Here, the value of the formal parameter $\formal{x}_{i,j}$ determines how often the loop $L_{i,j}$ is pumped. By Corollary \ref{cor:loop-pumping}, this corresponds to a word "inequality" in the parameters $\formal{x}_{i,j}$.

Note that there is one such "inequality" for each pair of "runs" $\rho_i,\rho_j$ with $0 \le i < j \le k$. 
By the choice of the intervals in $\bar{L}$, each of the "inequalities" is "satisfiable": 
indeed, the "inequality" for $\rho_i,\rho_j$ is "satisfied" by letting 
$\formal{x}_{i',j'} = 1$ if $i' \not = i$ or $j' \not = j$, and
$\formal{x}_{i,j} = 2$ otherwise. 

By Corollary~\ref{cor:saarela1}, the system of "inequalities" is also "satisfiable" with a tuple $\bar n = (n_{i,j})_{0 \le i < j \le k}$ of numbers, as claimed in the lemma.
\end{proof}

\section{The Decomposition Theorem}\label{sec:decomposition}

This section is devoted to the proof of the Decomposition Theorem:
\RestateGo{\restateDecomposition}
\Decomposition*

Our proof relies on the notion of 
\emph{"cover"} of an "SST", which is
reminiscent of the so-called ``lag-separation covering''
construction \cite{deSouza-thesis,SakarovitchDeSouzaMFCS08}.
Intuitively, given an SST $T$ and 
two integers $C,D \in \mathbb{N}$,
we construct an "SST" $\Cover{T}$
that is "equivalent" to $T$,
yet for each "input" $u$ it only admits 
pairs of "accepting" "runs" with 
different "outputs" or $C$-"delay" larger than $D$. 
The crucial point will be that $\Cover{T}$ is "$k$-ambiguous" when $T$ is "$k$-valued".

\begin{proposition}\label{prop:separation-covering}
Given an "SST" $T$ and two numbers $C,D$, 
one can compute an "SST" called
$""\Cover*[C,D]{T}@cover""$ such that
\begin{enumerate}
\item $\Cover{T}$ is "equivalent" to $T$;
\item for every two "accepting" "runs" $\rho\neq\rho'$ of $\Cover{T}$ 
      having the same "input", 
      either $\out{\rho} \neq \out{\rho'}$ or $\delay<C>{\rho,\rho'} >~D$;
\item every "accepting" "run" of $\Cover{T}$ can be projected onto an "accepting" "run" of $T$.
\end{enumerate}
\end{proposition}

\begin{proof}
We order the set of "accepting" "runs" of $T$ lexicographically,
and we get rid of all the "runs" for which there exists a 
lexicographically smaller "run" with the same "input", 
the same "output", and small "delay". 
Since all these conditions are encoded by regular languages, 
the remaining set of "runs" is also regular, and this can be 
used to construct an "SST" $\Cover{T}$ that satisfies the 
required properties.

We now give more details about this construction. 
Let $R$ denote the set of all "accepting" "runs" of $T$.
Remark that $R$ is a language over the alphabet consisting 
of transitions of $T$, and it is recognized by the underlying 
automaton of $T$, so it is regular.
\AP
Let
\[
    ""\Sep*{R} @separation""
    ~=~
        \big\{  \rho \in R
                ~\big|~
			    \nexists \rho' \in R \; .\,
			    \rho' \!<\! \rho ~\wedge~ \delay<C>{\rho,\rho'} \leq D
        \big\}.
\]
Recall that the delay is only defined for "accepting" "runs" 
with same "input" and same "output", 
so $\delay<C>{\rho,\rho'} \leq D$ implies that $\inp{\rho} = \inp{\rho'}$ 
and $\out{\rho} = \out{\rho'}$.
We show that
\begin{enumerate}[label={\sc\alph*)}]
    \item \label{item:sep1}
          $\Sep{R}$ is a regular subset of $R$;
    \item \label{item:sep2}
          $\{(\inp{\rho}, \out{\rho}) \:\mid\: \rho \in \Sep{R}\} = \{(\inp{\rho}, \out{\rho}) \:\mid\: \rho \in R\}$;
    \item \label{item:sep3}
          for every pair of "runs" $\rho, \rho' \in \Sep{R}$ over the same "input",
          either $\out{\rho} \neq \out{\rho'}$ or $\delay<C>{\rho,\rho'} > D$.
\end{enumerate}
Before proving these properties, let us show how to use them to conclude the 
proof of the proposition:

We start with a DFA $A$ recognizing $\Sep{R}$, whose existence is guaranteed by 
Property \ref{item:sep1}. Note that the transitions of $A$ are of the form
$\big(q,(s,a,\alpha,s'),q'\big)$, where $(s,a,\alpha,s')$ is a transition of $T$.
Without loss of generality, we assume that the source state $q$ 
of an $A$-transition determines the source state $s$ 
of the corresponding $T$-transition, and similarly
for the target states $q'$ and $s'$.
Thanks to this, we can turn $A$ into the desired "SST" $\Cover{T}$
by simply projecting away 
the $T$-states from the $T$-transitions, namely, by replacing every transition
$\big(q,(s,a,\alpha,s'),q'\big)$ with $(q,a,\alpha,q')$. 
To complete the construction, we observe that if the state $q'$
is final in $A$, then the corresponding state $s'$ is also final in $T$
(this is because $A$ recognizes only "accepting" "runs" of $T$). 
Accordingly, we can define the "final@@update" "update" of $\Cover{T}$
so that it maps any final state $q'$ of $A$ to the 
"final@@update" "update" $\finalupd(s')$, as determined by the corresponding
final state $s'$ in $T$.
Finally, thanks to Properties \ref{item:sep2} and \ref{item:sep3}, 
the "SST" $\Cover{T}$ constructed in this way clearly satisfies the 
properties claimed in the proposition.

Let us now prove Properties \ref{item:sep1}--\ref{item:sep3}.

\smallskip\noindent
{\sc Proof of Property \ref{item:sep1}.}~
Note that the set $\Sep{R}$ is obtained 
by combining the relations $R$, 
$\{ (\rho,\rho') ~|~ \rho' < \rho\}$, and
$\{ (\rho,\rho') ~|~ \delay<C>{\rho,\rho'} \leq~D\}$ 
using the operations of intersection, projection, and complement.
Also recall that $R$ can be regarded a regular language, and that
$\{ (\rho,\rho') ~|~ \rho' < \rho\}$ and
$\{ (\rho,\rho') ~|~ \delay<C>{\rho,\rho'}~\leq ~D\}$
are "automatic" relations (for the latter one uses Lemma~\ref{lem:regular-delay}).
It is also a standard result (cf.~\cite{Hodgson83,KhoussainovN95,BlumensathG00})
that "automatic" relations are closed under intersection, projection, and
complement.
From this it follows that $\Sep{R}$ is a regular language.

\smallskip\noindent
{\sc Proof of Property \ref{item:sep2}.}~
As $\Sep{R} \subseteq R$, the left-to-right inclusion is immediate.
To prove the converse inclusion,
consider an "input"-"output" pair $(u,v)$ in the right hand-side of the equation,
namely, $(u,v)$ is a pair in the relation realized by $T$.
Let $\rho$ be the lexicographically least "accepting" "run" of $T$
such that $\inp{\rho}=u$ and $\out{\rho}=v$.
By construction, $\rho\in\Sep{R}$ and hence $(u,v)$ also belongs to the
left hand-side of the equation.

\smallskip\noindent
{\sc Proof of Property \ref{item:sep3}.}~
This holds trivially by the definition of $\Sep{R}$.
\end{proof}

We can now present the missing ingredients of the decomposition result.
Proposition \ref{prop:coveringIsFiniteValued} below
shows that, for suitable choices of $C$ and $D$
that depend on the "valuedness" of $T$, 
$\Cover{T}$ turns out to be "$k$-ambiguous".

This will enable the decomposition result
via a classical technique that decomposes any 
"$k$-ambiguous" automaton/transducer into a 
union of $k$ "unambiguous" ones
(see Proposition \ref{prop:FiniteValuedToDecomposition}
further below).

\begin{proposition}\label{prop:coveringIsFiniteValued}
Let $T$ be a "$k$-valued" "SST" and let $C,D,m$ be
as in Lemma~\ref{lem:saarela-delay} (note that $m$ depends on $k$).
The "SST" $\Cover[Cm,Dm^2]{T}$ is "$k$-ambiguous".
\end{proposition}

\begin{proof}
We prove the contrapositive of the statement. Assume that 
$\Cover[Cm,Dm^2]{T}$ is not "$k$-ambiguous", that is, 
it admits $k+1$ "accepting" "runs" $\rho_0, \dots, \rho_k$ 
on the same "input". 
Recall from Proposition \ref{prop:separation-covering} 
that for all $0\le i<j\le k$, 
either $\out{\rho_i} \neq \out{\rho_j}$ or $\delay<Cm>{\rho_i,\rho_j} > Dm^2$.
By Lemma~\ref{lem:saarela-delay} we can find pumped versions
of the "runs" $\rho_0, \dots, \rho_k$ that have all the same "input" 
but have pairwise different "outputs", and thus $T$ is not "$k$-valued".
\end{proof}

\begin{proposition}\label{prop:FiniteValuedToDecomposition}
For all $k \in \mathbb{N}$, every "$k$-ambiguous" "SST" 
can be decomposed into a union of $k$ "unambiguous" "SSTs".
\end{proposition}

\begin{proof}
The decomposition is done via a classical 
technique applicable to "$k$-ambiguous" NFA and, 
by extension, to all variants of automata and transducers
(see~\cite{JohnsonICALP85,Sakarovitch98}).
More precisely, decomposing a "$k$-ambiguous" NFA into a union of $k$ 
"unambiguous" NFA is done by ordering runs lexicographically 
and by letting the $i$-th NFA in the decomposition guess
the $i$-th accepting run on a given input (if it exists).
Since the lexicographic order is a regular property of pairs of runs,
it is easy to track all smaller runs.
\end{proof}

\paragraph{Proof of Theorem~\ref{thm:FiniteValuedToDecomposition}.} 
We now have all the ingredients to prove
Theorem~\ref{thm:FiniteValuedToDecomposition}, which directly
follows from Propositions~\ref{prop:coveringIsFiniteValued}
and~\ref{prop:FiniteValuedToDecomposition}, and the fact that
"unambiguous" "SSTs" can be determinized~\cite{DBLP:conf/fsttcs/AlurC10}. 
\qed

\section{Finite valuedness}\label{sec:finite-valuedness}

We characterize "finite valuedness" of "SSTs" 
by excluding certain types of substructures.
Our characterization has strong analogies with 
the characterization
of "finite valuedness" for one-way transducers,
where the excluded substructures have the shape 
of a ``W'' and are therefore called \emph{W-patterns} 
(cf.~\cite{deSouza-thesis}).%
\footnote{\cite{deSouza-thesis} 
          used also other excluded substructures,
          but they can be seen as degenerate cases 
          of W-patterns.}

\begin{definition}\label{def:W-pattern}
A ""W-pattern"" is a substructure of an "SST" 
consisting of states $q_1,q_2,r_1,r_2,r_3$,
and some initial and final states, that are 
connected by "runs" as in the diagram of Figure~\ref{fig:Wpattern}.
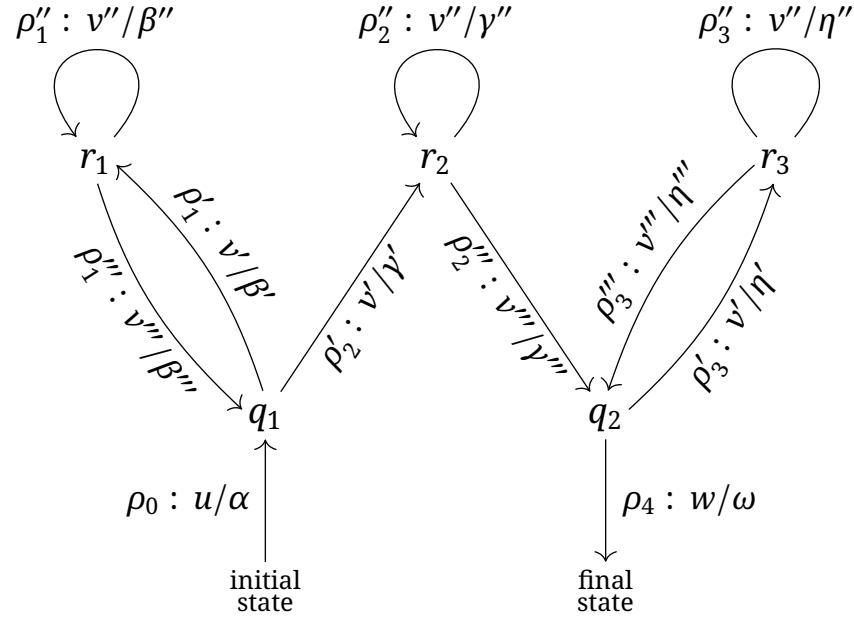
\begin{figure}[ht]
\begin{tikzpicture}[scale=1.5]
%
\node (q) at (0,1.75) {$q_1$};
\node (q') at (3,1.75) {$q_2$};
\node (r) at (-1.5,4) {$r_1$};
\node (s) at (1.5,4) {$r_2$};
\node (t) at (4.5,4) {$r_3$};
\node (initial) at (0,0.2) 
      {$\begin{subarray}{c} \text{initial} \\ \text{state} \end{subarray}$};
\node (final) at (3,0.2) 
      {$\begin{subarray}{c} \text{final} \\ \text{state} \end{subarray}$};
\draw (initial) edge [arrow] node [left=1mm] {\small $\rho_0:\, u/\alpha~$} (q);
\draw (q') edge [arrow] node [right=1mm] {\small $~\rho_4:\, w/\omega$} (final);
\draw (q.north) edge [arrow, bend right=15] 
      node [sloped, above=1mm] {\small $\rho'_1:\, v'/\beta'~~$} (r.east);
\draw (r) edge [arrow, loop above] 
      node [above=1mm] {\small $\rho''_1:\, v''/\beta''$} (r);
\draw (r.south) edge [arrow, bend right=15] 
      node [sloped, below=1mm] {\small $~\rho'''_1:\, v'''/\beta'''$} (q.west);
\draw (q) edge [arrow] 
      node [sloped, below=1mm] {\small $\rho'_2:\, v'/\gamma'\phantom{x}$} (s);
\draw (s) edge [arrow, loop above] 
      node [above=1mm] {\small $\rho''_2:\, v''/\gamma''$} (s);
\draw (s) edge [arrow] 
      node [sloped, below=1mm] {\small $\phantom{x}\rho'''_2:\, v'''/\gamma'''$} (q');
\draw (q'.east) edge [arrow, bend right=15] 
      node [sloped, below=1mm] {\small $\rho'_3:\, v'/\eta'~~$} (t.south);
\draw (t) edge [below, loop above] node [above=1mm] {\small $\rho''_3:\, v''/\eta''$} (t);
\draw (t.west) edge [arrow, bend right=15] 
      node [sloped, above=1mm] {\small $~~~\rho'''_3:\, v'''/\eta'''$} (q'.north);
\end{tikzpicture}
\caption{A "W-pattern".}
\label{fig:Wpattern}
\end{figure}

In that diagram a notation like $\rho: u'/\mu$ describes a run named $\rho$
that consumes an "input" $u'$ and produces an "update" $\mu$.
Moreover, the cyclic "runs" 
$\rho''_1$, 
$\rho''_2$,
$\rho''_3$,
$\rho'_1\rho''_1\rho'''_1$, 
$\rho'_2\rho''_2\rho'''_2$, and
$\rho'_3\rho''_3\rho'''_3$
are required to be "loops", namely,
their "updates" must have "idempotent" "skeletons".
\end{definition}

An important feature of the above definition is that 
the small "loops" at states $r_1,r_2,r_3$ consume the same "input", 
i.e.~$v''$, and, similarly, the big "loops" at $q_1$ and $q_2$, 
as well as the runs from $q_1$ to $q_2$, consume the same set 
of "inputs", i.e.~$v' \, (v'')^* \, v'''$. 

\AP
Given a "W-pattern" $P$ and a number $x \in \mathbb{N}_+$, 
we construct the following "runs" by pumping the small "loops" in the diagram of Figure \ref{fig:Wpattern} $x$ times:
\[
\begin{aligned}
	""\lrun*{x} @left run""  &~=~ \rho'_1 (\rho''_1)^x \, \rho'''_1 \\
    ""\mrun*{x} @mid run""   &~=~ \rho'_2 \, (\rho''_2)^x \, \rho'''_2 \\
    ""\rrun*{x} @right run"" &~=~ \rho'_3 \, (\rho''_3)^x \, \rho'''_3.
\end{aligned}
\]
\AP
Similarly, 
given a sequence 
$s=(x_1,x_2,\dots,x_{i-1},\underline{x_i},x_{i+1},\dots,x_n)$ 
of positive numbers with exactly one element underlined 
(we call such a sequence a ""marked sequence""), we define 
the "accepting" "run"  
\[
\begin{aligned}
	& ""\wrun*{s} @W-run"" ~~=~~ 
	\rho_0 ~~
    \underbrace{\lrun{x_1} ~ \lrun{x_2} ~\ldots~ \lrun{x_{i-1}}}_{\text{loops at $q_1$}} ~~
    \mrun{x_i} ~~
    \underbrace{\rrun{x_{i+1}} ~ \rrun{x_{i+2}} ~\ldots~ \rrun{x_n}}_{\text{loops at $q_2$}} 
    ~~ \rho_4.
\end{aligned}
\]
For each "marked sequence" 
$s=(x_1,\dots,x_{i-1},\underline{x_i},x_{i+1},\dots,x_n)$,
$\wrun{s}$ consumes the "input"   
\[
u \, v' (v'')^{x_1} v''' \, \dots \, 
      v' (v'')^{x_n} v''' \, 
 w
\]
and produces an "output" of the form
\[
\begin{aligned}
  \out{\wrun{s}} ~=~ 
  \big(\iota \, \alpha ~
  & \beta' (\beta'')^{x_1} \beta''' \,
  \dots \,
  \beta' (\beta'')^{x_{i-1}} \beta''' \,
  ~ \\
  & \gamma' (\gamma'')^{x_i} \gamma''' \,
  ~ \\
  & \eta' (\eta'')^{x_{i+1}} \eta''' \,
  \dots \,
  \eta' (\eta'')^{x_n} \eta'''
  ~ \omega \, \omega'\big)(X_1)
\end{aligned}
\]
where $\iota$ is the "initial@@update" "update"
and $\omega'$ is the "final@@update" "update"
determined by the final state of $\wrun{s}$.
\AP
Note that, differently from the output, the "input" 
only depends on the \reintro[marked sequence]{unmarked sequence},
and thus a "W-pattern" can have "accepting" "runs" that 
consume the same "input" and produce arbitrarily many 
different "outputs".
As an example, consider a "W-pattern" as in Definition \ref{def:W-pattern}, 
where $\gamma''$ is the only update that produces output symbols
-- say $\gamma''$ appends letter $c$ to the right of the unique variable. 
Further suppose that $u=w=\varepsilon$, $v'=v'''=a$, and $v''=b$. 
So, on "input" $(aba) \, (a b^2 a) \, \dots \, (a b^n a)$,
this "W-pattern" produces $n$ different "outputs":
$c$, $c^2$, \dots, $c^n$. 
The definition and the lemma below generalize this example.

\begin{definition}\label{def:skewed-W-pattern}
A "W-pattern" $P$ is ""divergent"" 
if there is a $5$-tuple of numbers $n_1,n_2,n_3,n_4,n_5\in\mathbb{N}_+$
for which the two "runs" $\wrun{(n_1,n_2,n_3,\underline{n_4},n_5)}$
and $\wrun{(n_1,\underline{n_2},n_3,n_4,n_5)}$ produce different outputs
(recall that the "runs" consume the same "input").
It is called ""simply divergent"" if in addition
$n_1,n_2,n_3,n_4,n_5\in\{1,2\}$.
\end{definition}

\begin{theorem}\label{thm:finite-valued}
An "SST" is "finite-valued" iff it does not admit a
"simply divergent" "W-pattern".
\end{theorem}

The two implications of the theorem are
shown in Sections \ref{subsec:left-to-right}
and \ref{subsec:right-to-left}; effectiveness
of the characterization is shown in the next
section.


\subsection{Effectiveness of finite valuedness}\label{subsec:effectiveness}

We show in this section a {\upshape\sc PSpace} decision procedure for the characterization of "finite valuedness" 
in terms of absence of "simply divergent" "W-patterns" 
(Theorem~\ref{thm:finite-valued}). 
We also prove that the "equivalence problem" for "deterministic" "SSTs", 
known to be in {\upshape\sc PSpace}, is polynomially reducible to the 
"finite valuedness" problem. Despite recent efforts by the community 
to better understand the complexity of the "equivalence problem", 
it is unknown whether the {\upshape\sc PSpace} upper bound for "equivalence" 
(and hence for "finite valuedness") can be improved, as no non-trivial 
lower bound is known. 
On the other hand, "equivalence" (as well as "finite valuedness") 
turns out to be in {\upshape\sc PTime} when the number of variables is 
fixed~\cite{DBLP:conf/icalp/AlurD11}.

Our effectiveness procedure uses the following complexity result on the
composition of deterministic "SSTs", which is of independent interest.
It is known
that deterministic "SSTs" are closed under
composition because of their equivalence to 
MSO transductions~\cite{DBLP:conf/fsttcs/AlurC10}. An automata-based construction for composition is given
in~\cite{DBLP:journals/corr/abs-2209-05448}. We show next an exponential construction 
based on reversible transducers~\cite{DBLP:conf/icalp/DartoisFJL17}.

\begin{proposition}\label{prop:compo}
Let $T_1$ and $T_2$ be two "deterministic" "SSTs" realizing the functions
$f_1 : \Sigma_1^*\rightarrow \Sigma_2^*$ and $f_2 :
\Sigma_2^*\rightarrow \Sigma_3^*$, respectively
Let $n_i$ (resp.~$m_i$) be the number of states  (resp.~variables) 
of $T_i$, and $M = n_1+n_2+m_1+m_2$. 
One can construct in time exponential in $M$ and 
polynomial in $|\Sigma_1|+|\Sigma_2|+|\Sigma_3|$ 
a "deterministic" "SST" realizing $f_1\circ f_2$, 
with exponentially many states and polynomially many
variables in $M$.
\end{proposition}

\begin{proof}
\AP
Each $T_i$ can be converted in
polynomial time into an "equivalent" 
two-way transducer that is ""reversible"",
i.e., both deterministic and co-deterministic
\cite{DBLP:conf/icalp/DartoisFJL17}. 
Reversible two-way transducers can be easily seen to be composable
in polynomial time~\cite{DBLP:conf/icalp/DartoisFJL17}, so we obtain a reversible two-way transducer $S$ 
realizing $f_1\circ f_2$, with state space polynomial
in $M$. Finally, it suffices to convert $S$ back to a "deterministic" "SST", 
which can be done in time exponential in the number of states of $S$ and polynomial 
in the size of the alphabets. This yields a "deterministic" "SST" with 
exponentially many states and polynomially many variables in $M$ 
(see e.g.~\cite{ledent13,DBLP:journals/ijfcs/DartoisJR18}).
\end{proof}

\RestateGo{\restateEffectiveness}
\Effectiveness*

\begin{proof}
    We start with an overview of the proof.
    By Theorem~\ref{thm:finite-valued}, it suffices to decide
    whether a given "SST" $T$ admits a "simply divergent"
    "W-pattern". Let us fix some tuple $\overline{x} =
    (x_1,\dots,x_5)\in\{1,2\}^5$. We construct an "SST"
    $T_{\overline{x}}$ which is \emph{not} "single-valued" iff $T$
    has a "W-pattern" which is "simply divergent" for $\overline{x}$. 
    Since checking "single-valuedness" of "SST" is 
    decidable~\cite{DBLP:conf/icalp/AlurD11}, 
    we can decide "finite valuedness" of $T$
    by solving "single-valuedness" problems
    for all "SST" $T_{\overline{x}}$, for all tuples
    $\overline{x}\in \{1,2\}^5$.
	Intuitively, we exhibit an encoding of
    "W-patterns" $P$ as words $u_P$,
    and show that the set of these encodings forms a regular language. 
    The encoding $u_P$ informally consists of the 
    "runs" that form the "W-pattern" $P$, and
    some of these "runs"
    are overlapped to be able to check that they are on the same "input". 
    Accordingly, the "SST" $T_{\overline{x}}$ will take as 
    "input" such an encoding $u_P$ and produce as "outputs" the 
    two words 
    $\out{\wrun{s}}$ and $\out{\wrun{s'}}$,
    where $s = (x_1,x_2,x_3,\underline{x_4},x_5)$ and 
    $s' = (x_1,\underline{x_2},x_3,x_4,x_5)$. 
	To achieve this, $T_{\overline{x}}$ can consume the "input" $u_P$
    while iterating the encoded "runs" as prescribed by $s$ or $s'$
    and simulating the transitions to construct the appropriate outputs.
    Finally, an analysis of the size of $T_{\overline{x}}$ and of the 
    algorithm from~\cite{DBLP:conf/icalp/AlurD11} 
    for checking "single-valuedness" gives the {\upshape\sc PSpace} upper bound.

    \paragraph{Detailed reduction} 
    We now explain in detail the reduction to
    the "single-valuedness" problem. We will 
    then show how to derive the {\upshape\sc PSpace} upper bound 
    by inspecting the decidability proof for "single-valuedness". 

    Let $T = (\alp, \var, Q, Q_{\mathrm{init}}, Q_{\mathrm{final}}, \finalupd, \Delta)$ 
    be the given "SST" and 
    let $\mathcal{P}$ be the set of all "W-patterns" of $T$.
    We first show that $\mathcal{P}$ is a regular set, modulo some 
    well-chosen encodings of "W-patterns" as words. 
    Recall that a "W-pattern" consists of a tuple of runs
    $P = (\rho_0,
          \rho'_1,\rho''_1,\rho'''_1,
          \rho'_2,\rho''_2,\rho'''_2,
          \rho'_3,\rho''_3,\rho'''_3,
          \rho_4)$,
    connected as in the diagram of
    Definition~\ref{def:W-pattern}. 
    Note that some of those runs share a common "input" 
    (e.g.~$\rho''_1,\rho''_2,\rho''_3$ share the "input" $v''$). 
    Therefore, we cannot simply encode $P$ as a sequence 
    of "runs" $\rho_0,\rho'_1,\dots$, as otherwise regularity 
    would be lost. 
    Instead, in the encoding we overlap groups of "runs" 
    over the same "input", precisely, 
    the group $\{\rho'_1,\rho'_2,\rho'_3\}$ on "input" $v'$,
    the group $\{\rho''_1,\rho''_2,\rho''_3\}$ on "input" $v''$,
    and the group $\{\rho'''_1,\rho'''_2,\rho'''_3\}$ on "input" $v'''$.
	Formally, this is done by taking the "convolution"
	of the "runs" in each group, which results
    in a word over the alphabet $\Delta^3$. 
    Accordingly, $P$ is encoded as the word
    \[
      u_P ~=~ 
      \rho_0 \:\#\: 
      (\rho'_1 \otimes \rho'_2 \otimes \rho'_3) \:\#\:
      (\rho''_1 \otimes \rho''_2 \otimes \rho''_3) \:\#\:
      (\rho'''_1 \otimes \rho'''_2 \otimes \rho'''_3) \:\#\:
      \rho_4
    \]
	where $\#$ is a fresh separator. 
	The language $L_{\mathcal{P}} = \{ u_P \mid P\in\mathcal{P}\}$, 
	consisting of all encodings of "W-patterns", is easily seen to
        be regular, recognizable by some automaton $A_{\mathcal{P}}$
        which checks that runs forming each "convolution" share the
        same "input" and verify "skeleton idempotency" 
        (recall that "skeletons" form a finite monoid). 
	The number of states of the automaton $A_{\mathcal{P}}$ 
	turns out to be polynomial in the number of states of $T$
	and in the size of the "skeleton monoid", which
	in turn is exponential in the number of variables. 
    
    Next, we construct the "SST" $T_{\overline x}$ 
    as the disjoint union of two
    deterministic "SSTs" $T_s$ and $T_{s'}$, 
    where $s = (x_1,x_2,x_3,\underline{x_4},x_5)$ and 
    $s' = (x_1,\underline{x_2},x_3,x_4,x_5)$. 
    We only describe $T_s$, as the construction of 
    $T_{s'}$ is similar. 
    The "SST" $T_s$ is obtained as a suitable restriction of 
    the composition of two "deterministic" "SSTs" 
    $T_{\mathrm{iter}}^s$ and $T_{\mathrm{exec}}$, 
    which respectively iterate the "runs" as 
    prescribed by $s$ and execute the 
    transitions read as "input". 
    When fed with the encoding $u_P$ of a "W-pattern", 
    $T_{\mathrm{iter}}^s$ needs to output $\wrun{s}\in\Delta^*$.
    More precisely, it takes as "input" a word of the form 
	$
	  \rho_0 \:\#\: 
	  (\rho'_1 \otimes \rho'_2 \otimes \rho'_3) \:\#\:
	  (\rho''_1 \otimes \rho''_2 \otimes \rho''_3) \:\#\:
	  (\rho'''_1 \otimes \rho'''_2 \otimes \rho'''_3) \:\#\:
	  \rho_4
    $
    and produces as "output"
    \[
    \begin{aligned}
      \rho_0 ~~
      & \rho'_1 \, (\rho''_1)^{x_1} \, \rho'''_1 ~~ 
        \rho'_1 \, (\rho''_1)^{x_2} \, \rho'''_1 ~~ 
        \rho'_1 \, (\rho''_1)^{x_3} \, \rho'''_1 ~~ \\
      & \rho'_2 \, (\rho''_2)^{x_4} \, \rho'''_2 ~~ \\
      & \rho'_3 \, (\rho''_3)^{x_5} \, \rho'''_3 ~~ \rho_4.
	\end{aligned}
    \]

    The SST $T_{\mathrm{iter}}^s$ uses one variable for each
    non-iterated run (e.g.~for $\rho_0$ and $\rho'_1$),
    $x_1+x_2+x_3$ variables to
    store copies of $\rho''_1$, $x_4$ variables to store copies of
    $\rho''_2$, and $x_5$ variables to store copies of $\rho''_3$, and
    eventually outputs the concatenation of all these variables to obtain
    $\wrun{s}$. Note that $T_{\mathrm{iter}}^s$ does not 
	need to check that the "input" is a well-formed encoding (this is done later when constructing $T_s$), 
	so the number of its states and variables is bounded
	by a constant;
	on the other hand, the input alphabet, consisting
	of transitions of $T$, is polynomial in the size of $T$.
        
    The construction of $T_{\mathrm{exec}}$ is straightforward:
    it just executes the transitions it reads along the "input", 
    thus simulating a run of $T$. 
    Hence $T_{\mathrm{exec}}$ has a single state and the same 
    number of variables as $T$. Its alphabet is linear in 
    the size of $T$. 
	
    Now, $T_s$ is obtained from the  composition
    $T_{\mathrm{exec}} \circ T_{\mathrm{iter}}^s$ by 
    restricting the "input" domain to $\mathcal{P}$.
    It is well-known that
    "deterministic" "SST" are closed under  composition 
    and regular domain restriction~\cite{DBLP:conf/fsttcs/AlurC10}. 
    By the above constructions, we have
    \[
      T_{\bar x}(u_P) 
      ~=~ T_s(u_P) ~\cup~ T_{s'}(u_P)
      ~=~ \{\out{\wrun{s}},\out{\wrun{s'}}\}
    \]
    and hence $T$ contains a "W-pattern" that is 
    "simply divergent" for $\bar x$ iff 
    $T_{\bar x}$ is not "single-valued".
    This already implies the decidability of the
    existence of a "simply divergent" "W-pattern" in $T$,
    and hence by Theorem \ref{thm:finite-valued},
    of "finite valuedness".

    \paragraph{Complexity analysis} 
    Let us now analyse the complexity  in detail.    
    This requires first estimating the size of $T_{\bar x}$. Let $n_T$ resp. $m_T$ be the number of states of $T$,
    resp.~its number of variables. 
   From the previous bounds on the sizes of
    $T_{\mathrm{exec}}$ and $T^s_{\mathrm{iter}}$ and Proposition~\ref{prop:compo}, 
    we derive that the number of states and variables of 
    $T_{\mathrm{exec}}\circ T^s_{\mathrm{iter}}$ is polynomial 
    in both $n_T$ and $m_T$.
    Further, restricting the domain to $\mathcal{P}$ is done 
    via a product with the automaton $A_\mathcal{P}$,
    whose size is polynomial in $n_T$ 
    and exponential in $m_T$.
    Summing up, the number of states of $T_s$ is exponential 
    in $m_T$, and polynomial in $n_T$.
    Its number of variables is polynomial in both $n_T$ and $m_T$.
    And so do $T_{s'}$ and $T_{\bar x}$. 
    
    As explained in~\cite{DBLP:conf/icalp/AlurD11}, 
    checking "single-valuedness" of "SST" reduces to 
    checking non-emptiness of a 1-reversal 2-counter 
    machine of size exponential in the number of variables 
    and polynomial in the number of states. 
    This is  fortunate, since it allows us
    to conclude that checking "single-valuedness" of 
    $T_{\bar x}$ reduces to checking non-emptiness of 
    a 1-reversal 2-counter machine of size just exponential 
    in the number of variables of $T$. 
    The {\upshape\sc PSpace} upper bound (and the {\upshape\sc PTime} upper
    bound for a fixed number of variables) now follow by recalling
    that non-emptiness of counter machines with fixed numbers of reversals and counters is
    in {\upshape\sc NLogSpace}~\cite{DBLP:journals/jcss/GurariI81}.

    \paragraph{Lower bound} For the lower bound, consider two "deterministic" "SST" $T_1,T_2$ over
    some alphabet $\alp$ with same domain $D$. Domain equivalence can be
    tested in {\upshape\sc PTime} because $T_1,T_2$ are "deterministic". 
    Consider a fresh symbol $\#\not\in\Sigma$,
    and the relation
    \[
    R ~=~ \Big\{ \big(u_1 \# \dots \# u_n, \: T_{i_1}(u_1) \# \dots \# T_{i_n}(u_n)\big) 
                 ~\Big|~
                 \begin{smallmatrix}
                 u_i\in D, ~
                 n\in\mathbb{N}, \\
                 i_1,\dots,i_n\in\{1,2\} 
                 \end{smallmatrix} 
          \Big\}
    \]
    It is easily seen that $R$ is realizable 
    by a (non-deterministic) "SST". 
    We claim that $R$ is "finite-valued" iff it is "single-valued", iff
    $T_1$ and $T_2$ are "equivalent". 
    If $T_1$ and $T_2$ are "equivalent",
    then $T_1(u_j) = T_2(u_j)$ for all $1\leq j\leq n$, hence $R$ is
    "single-valued", and so "finite-valued". 
    Conversely, if $T_1$ and $T_2$ are not "equivalent", then 
    $T_1(u)\neq T_2(u)$ for some $u\in D$, and the family of "inputs" 
    $(u\#)^n \, u$, with $n\in\mathbb{N}$, witnesses the fact that
    $R$ is not "finite-valued". 
\end{proof}

\begin{remark}
    The proof of the previous theorem can be adapted to show that the "equivalence" problem for deterministic "SSTs" and the "finite valuedness" problem for "SSTs" are equivalent (modulo polynomial many-one reductions).
\end{remark}

As a corollary, we obtain an alternative proof of the following 
known result:

\RestateGo{\restateYen}
\Yen*

    \begin{proof}
      Observe that a necessary condition for a two-way transducer to
      be finite-valued is that crossing sequences are bounded.
      More precisely, if a crossing sequence has a loop then the
      output of the loop must be empty, otherwise the transducer is
      not finite valued.
      Given a bound on the length of crossing sequences the standard
      conversion into an equivalent "SST" applies, see
      e.g.~\cite{ledent13,DBLP:journals/ijfcs/DartoisJR18}. This
      yields an SST with an exponential number of states and a linear
      number of variables, both in the number of
      states of the initial two-way
      transducer. Finally, we apply the algorithm of
      Theorem~\ref{thm:effectiveness}, and we observe that it amounts
      to checking emptiness of a $1$-reversal $2$-counter machine
      whose number of states is exponential in the number of states of
      the initial two-way transducer. We  conclude again by applying the
      {\upshape\sc NLogSpace} algorithm for checking emptiness of such
      counter machines~\cite{DBLP:journals/jcss/GurariI81}. 
    \end{proof}

\subsection{A necessary condition for finite valuedness}\label{subsec:left-to-right}

Here we prove the contrapositive of the left-to-right
implication of Theorem \ref{thm:finite-valued}:
we show that a
"divergent" "W-pattern" can generate arbitrarily
many "outputs" on the same "input".

\begin{lemma}\label{lem:unbounded-values1}
Every "SST" that contains some "divergent" "W-pattern" is not "finite-valued".
\end{lemma}

\begin{proof}
Let us fix an "SST" with a "divergent" "W-pattern" $P$.
In order to prove that the "SST" is not "finite-valued", we show that we can construct
arbitrary many "accepting" "runs" 
of $P$ that consume the 
same "input" and produce pairwise different "outputs".
To do this we will consider for some suitable $M \in \mathbb{N}$
"inequalities" in the formal parameters 
$\formal{s}_1, \ldots, \formal{s}_M \in \mathbb{N}_+$ (where $1 \leq i < j \leq M$),
and look for arbitrary large, "satisfiable", sets of "inequalities" of the form:
\begin{align*}
    &e_{M,i,j}[\formal{s}_1,\formal{s}_2, \ldots, \formal{s}_M]: 
    \quad\left.
    \begin{array}{ccc}
         & \out{\wrun{\formal{s}_1,\formal{s}_2,\ldots,\formal{s}_{i-1},\underline{\formal{s}_{i}},\formal{s}_{i+1}, \ldots ,\formal{s}_M}}\\
         & \neq\\
         & \out{\wrun{\formal{s}_1,\formal{s}_2,\ldots,\formal{s}_{j-1},\underline{\formal{s}_{j}},\formal{s}_{j+1}, \ldots ,\formal{s}_M}}.
    \end{array}
    \right.
\end{align*}
Recall that, according to the diagram of Definition \ref{def:W-pattern},
the number of variable occurrences before (resp.~after) the underlined parameter 
represents the number of "loops" at state $q_1$ (resp.~$q_2$) in a "run" of the "W-pattern".
Moreover, each variable $\formal{s}_i$ before (resp.~after) the underlined parameter 
represents the number of repetitions of the small "loops" at $r_1$ (resp.~$r_3$) within 
occurrences of bigger "loops" at $q_1$ (resp.~$q_2$); similarly, the underlined
variable represents the number of repetitions of the "loop" at $r_2$ within
the "run" that connects $q_1$ to $q_2$.
In view of this, by Corollary~\ref{cor:loop-pumping}, the "outputs" of the "runs" 
considered in the above inequality have the format required for a "word inequality" 
with repetitions parametrized by $\formal{s}_1, \ldots, \formal{s}_M$.

The fact that the "W-pattern" $P$ is "divergent" will help to find sets
of "satisfiable" "inequalities" $e_{M,i,j}$ of arbitrary large
cardinality. This, in turn, will produce (combined with our word
combinatorics results) arbitrary many "accepting" "runs" 
over the 
same "input", having pairwise different "outputs".

\begin{claim*}
    For every $m \in \mathbb{N}$, there exist $M \in \mathbb{N}$ and a set $I \subseteq \{1,2,\ldots,M\}$ 
    of cardinality $m + 1$ such that, for all $i<j\in I$, $e_{M,i,j}$ is "satisfiable".
\end{claim*}

\begin{proof}[Proof of the claim]
Since $P$ is a "divergent" "W-pattern", there exist
$n_1,n_2,n_3,n_4,n_5 \in \mathbb{N}_+$ such that
\[
  \out{\wrun{n_1,\underline{n_2},n_3,n_4,n_5}}
  ~\neq~
  \out{\wrun{n_1,n_2,n_3,\underline{n_4},n_5}}. 
\]
We fix such numbers $n_1,n_2,n_3,n_4,n_5 \in \mathbb{N}_+$.
Consider now the following "inequality" over the formal parameters 
$\formal{x}, \formal{y}, \formal{z}$:
\begin{align*}
    &e[\formal{x}, \formal{y}, \formal{z}]: 
    \quad\left.
    \begin{array}{ccc}
         & \out{\wrun{
         \overbrace{n_1,n_1,\ldots,n_1}^{\formal{x} \textup{ times}}, \: \underline{n_2}, \:
         \overbrace{n_3,\ldots,n_3}^{\formal{y} \textup{ times}}, \: n_4, \:
         \overbrace{n_5,\ldots,n_5}^{\formal{z} \textup{ times}}}}\\[1ex]
         & \neq\\
         & \out{\wrun{
         \underbrace{n_1,n_1,\ldots,n_1}_{\formal{x} \textup{ times}}, \: n_2, \:
         \underbrace{n_3,\ldots,n_3}_{\formal{y} \textup{ times}}, \: \underline{n_4}, \:
         \underbrace{n_5,\ldots,n_5}_{\formal{z} \textup{ times}}}}.
    \end{array}
    \right.
\end{align*}
Note that every instance of $e[\formal{x}, \formal{y}, \formal{z}]$ 
with concrete values $x,y,z$ is also an instance of $e_{M,i,j}$, where 
$M=x+y+z+2$, $i=x+1$, $j=x+y+2$, and all parameters 
$\formal{s}_1, \ldots, \formal{s}_M$ are instantiated with values from $\{n_1, \ldots, n_5\}$. 
Moreover, as the parameters in $e[\formal{x}, \formal{y}, \formal{z}]$ determine the number of 
repetitions of $n_1,n_3,n_5$, which in their turn correspond to "pumping" 
"loops" at $q_1$ and $q_2$,
by Corollary~\ref{cor:loop-pumping}, the "outputs" of the considered "runs" 
have the format required for a "word inequality" with repetitions 
parametrized by $\formal{x}, \formal{y}, \formal{z}$.

Since $e[\formal{x}, \formal{y}, \formal{z}]$ is "satisfiable" (e.g.~with
$\formal{x}=\formal{y}=\formal{z}=1$), Corollary~\ref{cor:saarela2}
implies that
\[
\begin{aligned}
  & \exists \ell_{\formal{y}} ~ \forall h_{\formal{y}} ~ 
  \exists \ell_{\formal{x}} ~ \forall h_{\formal{x}} ~ 
  \exists \ell_{\formal{z}} ~ \forall h_{\formal{z}} \\
  &\qquad\quad
  \underbrace{[\ell_{\formal{x}},h_{\formal{x}}]}_{\text{values for }\formal x} \times
  \underbrace{[\ell_{\formal{y}},h_{\formal{y}}]}_{\text{values for }\formal y} \times
  \underbrace{[\ell_{\formal{z}},h_{\formal{z}}]}_{\text{values for }\formal z} 
  ~\subseteq~ \sol{e}.
\end{aligned}
\]
Note that we start by quantifying over $\ell_y$ and not $\ell_x$
(Corollary~\ref{cor:saarela2} is invariant with respect to the
parameter order). 
exist three integers $\ell_y,\ell_x,\ell_z>0$ such that 
\begin{equation}\label{equ:solutionXYZ}
[\ell_{\formal{x}}, \ell_{\formal{x}} + 2m\ell_{\formal{y}}] \times
[\ell_{\formal{y}}, 2m \ell_{\formal{y}}] \times
[\ell_{\formal{z}}, \ell_{\formal{z}} + 2m \ell_{\formal{y}}]
  ~\subseteq~ \sol{e}.
\end{equation}
Note that $h_y = 2m\ell_y$ depends only on $\ell_y$, while
$h_x=\ell_x+2m\ell_y$ depends on both $\ell_x$ and~$\ell_y$. 

We can now prove the claim by letting 
$M = \ell_{\formal{x}} + 2m\ell_{\formal{y}} + \ell_{\formal{z}} + 1$ and 
$I = \{\ell_{\formal{x}} + 2\lambda\ell_{\formal{y}} + 1 \mid 0 \leq \lambda \leq m\}$.
The gap between two consecutive values of $I$ equals $2\ell_y$,
and for every $i < j\in I$ we get
\[
\begin{array}{lllll}
    i-1 & \in &
  [\ell_{\formal{x}},\ell_{\formal{x}}+2(m-1)\ell_{\formal{y}}] & \subseteq &
  [\ell_{\formal{x}}, \ell_{\formal{x}} + 2m\ell_{\formal{y}}] \\
  j-i-1 & \in & [2\ell_{\formal{y}}-1,2m\ell_{\formal{y}}-1] & \subseteq &
  [\ell_{\formal{y}},2m\ell_{\formal{y}}] \\
  M-j & \in &
  [\ell_{\formal{z}}, \ell_{\formal{z}} + 2(m-1)\ell_{\formal{y}}] & \subseteq &
  [\ell_{\formal{z}}, \ell_{\formal{z}} + 2m \ell_{\formal{y}}].
\end{array}
\]

Thus, by Equation~\eqref{equ:solutionXYZ},
$(i - 1, j - i - 1, M - j) \in \sol{e}$. 
This "solution" of $e$ corresponds to the instance of  $e_{M,i,j}$ 
with the values for the formal parameters $\formal{s}_1,\dots,\formal{s}_M$
defined by

\[
\formal{s}_h = \left\{
\begin{array}{lll}
n_1 & \textup{ for every } 1 \leq h \leq i-1,\\
n_2 & \textup{ for } h = i,\\
n_3 & \textup{ for every } i+1 \leq h \leq j-1,\\
n_4 & \textup{ for } h = j,\\
n_5 & \textup{ for every } j+1 \leq M.
\end{array}
\right.
\]
Hence, $e_{M,i,j}$ is "satisfiable" for all $i<j \in I$, as claimed.
\end{proof}

We can now conclude the proof of the lemma using the above claim:
Corollary~\ref{cor:saarela1} tells us that any "system" of "word inequalities" 
is "satisfiable" when every "word inequality" in it is so.
Using this and the above claim, we derive that for every $m$ there exist 
$t_1,t_2,\ldots,t_M \in \mathbb{N}_+$ such that, for all $i<j\in I$ (with $I$ as in the claim),
$e_{M,i,j}[t_1,t_2,\ldots,t_M]$ holds.
For every $h \in I$, let
\[
    \rho_h ~=~ 
    \wrun{t_1,t_2,\ldots,t_{h-1},\underline{t_{h}},t_{h+1}, \ldots ,t_M}.
\]
Note that all runs $\rho_h$, for $h \in I$, consume the same "input",
since they all correspond to the same unmarked sequence
$(t_1,t_2,\ldots, t_M)$.
However, they produce pairwise different "outputs", because 
for every $i<j \in I$, the tuple $(t_1,t_2,\ldots,t_M)$
is a "solution" of $e_{M,i,j}$. 
Since $|I|=m$ can be chosen arbitrarily, the transducer is not "finite-valued". 
\end{proof}

\subsection{A sufficient condition for finite valuedness}
\label{subsec:right-to-left}

We finally prove that  any "SST" that exhibits  no 
"simply divergent" "W-pattern" is "finite-valued".
The proof relies on two crucial results.
The first one is a characterization of "finite ambiguity"
for "SSTs", which is easily derived from  the characterization
of "finite ambiguity" for finite state automata
\cite{finite-ambiguity1977simon,finite-ambiguity1977jacob,seidl91,AllauzenMR11}:

\begin{definition}\label{def:dumbbell}
A ""dumbbell"" is a substructure of an "SST"
consisting of states $q_1,q_2$ connected by 
"runs" as in the diagram
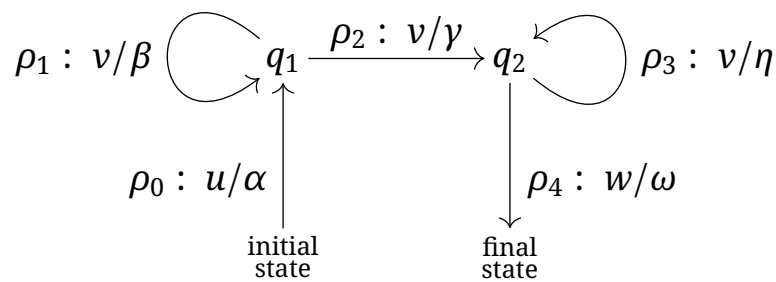
\[
\begin{tikzpicture}[baseline=47, scale=1.5]
	\myclip{-2.75,-0.5}{4.75,2.5};
	\node (q) at (0,1.75) {$q_1$};
	\node (q') at (2,1.75) {$q_2$};
	\node (initial) at (0,0) 
	      {$\begin{subarray}{c} \text{initial} \\ \text{state} \end{subarray}$};
	\node (final) at (2,0) 
	      {$\begin{subarray}{c} \text{final} \\ \text{state} \end{subarray}$};
	\draw (initial) edge [arrow] node [below left=1mm] {$\rho_0: ~u/\alpha~$} (q);
	\draw (q') edge [arrow] node [below right=1mm] {$~\rho_4: ~w/\omega$} (final);
	\draw (q) edge [arrow, loop left] node [left=2mm] {$\rho_1: ~v/\beta$} (q);
	\draw (q) edge [arrow] node [above=1mm] {$\rho_2: ~v/\gamma$} (q');
	\draw (q') edge [arrow, loop right] node [right=2mm] {$\rho_3: ~v/\eta$} (q');
\end{tikzpicture}
\]
where
the "runs" $\rho_1$ and $\rho_3$ are "loops" (in particular,
they produce "updates" with "idempotent" "skeletons") and
at least two among the runs $\rho_1,\rho_2,\rho_3$ 
are distinct.
\end{definition}

\begin{lemma}\label{lem:finite-ambiguity}
An "SST" is "finite-ambiguous" iff 
it does not contain any "dumbbell".
\end{lemma}

\begin{proof}
Let 
$T = (\alp, \var, Q, Q_{\mathrm{init}}, Q_{\mathrm{final}}, \finalupd, \Delta)$
be an "SST".
\AP
By projecting away the "updates" on the transitions we obtain from $T$ a 
""multiset finite-state automaton"" $A$.
Formally, $A = (\alp, Q,Q_{\mathrm{init}}, Q_{\mathrm{final}}, \Delta')$,
where $\Delta'$ is the \emph{multiset} containing one occurrence of a
triple $(q,a,q')$ for each transition of the form $(q,a,\alpha,q')$ in $\Delta$.
Note that a "multiset automaton" can admit several occurrences of 
the same ("accepting") "run".
Accordingly, the notion of "finite ambiguity" for $A$ requires the existence
of a uniform bound on the number of \emph{occurrences} of "accepting" "runs"
of $A$ on the same "input".
We also remark that "multiset automata" are essentially the same 
as weighted automata over the semiring of natural numbers (the weight of a transition being its number of occurrences), with only a difference in terminology where
"finite ambiguity" in "multiset automata" corresponds to 
"finite valuedness" in weighted automata.

Given the above construction of $A$ from $T$, one can verify by 
induction on $|u|$ that the number of occurrences of "accepting" "runs"
of $A$ on $u$ coincides with the number of "accepting" "runs" of $T$
on $u$. This means that $A$ is "finite-ambiguous" iff 
$T$ is "finite-ambiguous".

\AP
Finally, we recall the characterizations of "finite ambiguity"
from \cite{finite-ambiguity1977simon,finite-ambiguity1977jacob,seidl91}
(see in particular Theorem 1.1 and Lemma 2.6 from 
\cite{finite-ambiguity1977simon}).
In short, their results directly imply 
that a "multiset automaton" is "finite-ambiguous" iff 
it does not contain a ""plain dumbbell"", namely, a substructure
of the form
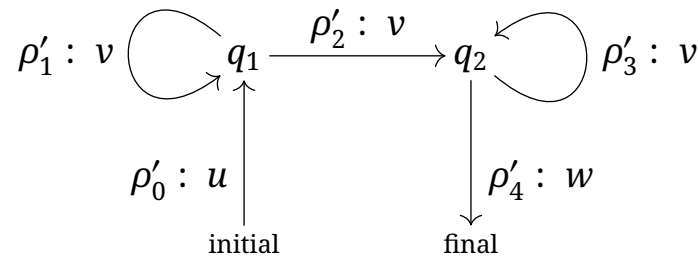
\[
\begin{tikzpicture}[baseline=47, scale=1.5]
	\myclip{-2.5,0}{4.5,2.5};
	\node (q) at (0,1.75) {$q_1$};
	\node (q') at (2,1.75) {$q_2$};
	\node (initial) at (0,0) 
	      {$\begin{subarray}{c} \text{initial} \\ \text{state} \end{subarray}$};
	\node (final) at (2,0) 
	      {$\begin{subarray}{c} \text{final} \\ \text{state} \end{subarray}$};
	\draw (initial) edge [arrow] node [below left=1mm] {$\rho'_0: ~u~$} (q);
	\draw (q') edge [arrow] node [below right=1mm] {$~\rho'_4: ~w$} (final);
	\draw (q) edge [arrow, loop left] node [left=2mm] {$\rho'_1: ~v$} (q);
	\draw (q) edge [arrow] node [above=1mm] {$\rho'_2: ~v$} (q');
	\draw (q') edge [arrow, loop right] node [right=2mm] {$\rho'_3: ~v$} (q');
\end{tikzpicture}
\]
where at least two among $\rho'_1,\rho'_2,\rho'_3$
are distinct  "runs".

This almost concludes the proof of the lemma, since any "dumbbell" 
of $T$ can be projected into a "plain dumbbell" of $A$.
The converse implication, however, is not completely straightforward.
The reason is that the cyclic "runs" of a "plain dumbbell" in $A$
do not necessarily correspond to "loops" in the "SST" $T$, 
as the "runs" need not produce "updates" 
with "idempotent" "skeletons".
Nonetheless, we can reason as follows.
Suppose that $A$ contains a "plain dumbbell", 
with occurrences of "runs" $\rho'_0,\rho'_1,\rho'_2,\rho'_3,\rho'_4$ 
as depicted above.
Let $\rho_0,\rho_1,\rho_2,\rho_3,\rho_4$ be some
corresponding "runs" in $T$ (with $\rho_i$ projecting to $\rho'_i$) and let 
$\alpha,\beta,\gamma,\eta,\omega$ be their induced "updates".
Further let $n$ be a large enough number such that
$\beta^n$ and $\eta^n$ have "idempotent" "skeletons"
(such an $n$ always exists since the "skeleton monoid"
is finite).
Now consider the substructure in $T$ given by the "runs" 
$\rho_0$, 
$(\rho_1)^n$, 
$(\rho_1)^{n-1} \, \rho_2$, 
$(\rho_3)^n$, 
and
$\rho_4$. 
This substructure satisfies precisely the definition 
of "dumbbell" for the "SST" $T$.
\end{proof}

%

The second ingredient for the proof of the right-to-left 
implication of Theorem \ref{thm:finite-valued}
uses once more the cover construction 
described in Proposition \ref{prop:separation-covering}.
More precisely, in Lemma \ref{lem:unbounded-values3} below
we show that if an "SST" $T$ has no "simply divergent" "W-pattern", 
then, for some well chosen values $C$, $D$, $m$, 
the "SST" $\Cover*[Cm,Dm^2]{T}$ contains no "dumbbell".
Before proving the lemma, let us show how it can be used 
to establish the right-to-left implication of Theorem \ref{thm:finite-valued}.

\paragraph{Proof of Theorem~\ref{thm:finite-valued}}
By Lemma~\ref{lem:finite-ambiguity}, 
if $\Cover*[Cm,Dm^2]{T}$ has no "dumbbell",
then it is "finite-ambiguous", hence 
"finite-valued".
Since $\Cover*[Cm,Dm^2]{T}$  and $T$ are "equivalent",
$T$ is "finite-valued", too.
\qed

\begin{lemma}\label{lem:unbounded-values3}
Given an SST $T$, one can compute  numbers $C$, $D$, $m$ such that
if $\Cover*[Cm,Dm^2]{T}$ contains a "dumbbell", then $T$ contains a 
"simply divergent" "W-pattern".
\end{lemma}

\begin{proof}
We first provide some intuition. 
If $\Cover[Cm,Dm^2]{T}$ contains a "dumbbell" for suitable values of $C,D,m$, then we show that
this "dumbbell" admits two distinct "runs" $\pi,\pi'$
that have either different "outputs" or  large "delay". 
In both cases we will be able to transform the "dumbbell" 
into a "simply divergent" "W-pattern" in $\Cover[Cm,Dm^2]{T}$, hence $T$ will have one as well.

Formally, let $T$ be an "SST". Let $C,D$ be defined 
as in Lemma~\ref{lem:delay-vs-output}, and
$m = 7E^{H^2 + H + 1} + 1$, where $E,H$ are 
defined as in Lemma~\ref{lem:ramsey}.
Next, suppose that $\Cover[Cm,Dm^2]{T}$ contains
a "dumbbell" as in Definition \ref{def:dumbbell}, 
with "runs" $\rho_0,\rho_1,\rho_2,\rho_3,\rho_4$
that produce respectively the "updates"
$\alpha,\beta,\gamma,\eta,\omega$. 

Consider the following "accepting" "runs", which are obtained
by composing the copies of the original "runs" of the "dumbbell". The runs $\pi, \pi'$ are different because at least two of $\rho_1,\rho_2,\rho_3$ are different:
\begin{equation}\label{eq:cover-runs}
\begin{aligned}
	\pi  &~=~  \rho_0 ~ \rho_1 ~ \rho_2 ~ \rho_3 ~ \rho_3 ~ \rho_3 ~ \rho_4\\
    \pi' &~=~ \rho_0 ~ \rho_1 ~ \rho_1 ~ \rho_1 ~ \rho_2 ~ \rho_3 ~ \rho_4.
\end{aligned}
\end{equation}
By the properties of $\Cover[Cm,Dm^2]{T}$, 
since $\pi$ and $\pi'$ consume the same "input",
they either produce different "outputs" or
have $Cm$-"delay" larger than $Dm^2$.

We first consider the case where the "outputs" are 
different. In this case, we can immediately
witness a "simply divergent" "W-pattern" $P$
by adding empty "runs" to the "dumbbell";
formally, for every $i=1,2,3$, we let 
$\rho'_i = \rho_i$ and $\rho''_i = \rho'''_i = \varepsilon$,
so as to form a "W-pattern" like the one in Figure~\ref{fig:Wpattern}, with
$r_1 = q_1$ and $r_2 = r_3 = q_2$.
Using the notation introduced at the beginning 
of Section \ref{sec:finite-valuedness}, we observe that 
$\pi  = \wrun{1,\underline{1},1,1,1}$ and 
$\pi' = \wrun{1,1,1,\underline{1},1}$
---
recall that the underlined number represents
how many times the small "loop" at $r_2$,
which is empty here, is repeated along the
"run" from $q_1$ to $q_2$, and the
other numbers represent how many times the small
"loops" at $r_1$ and $r_3$, which are also empty here,
are repeated within the occurrences of big "loops" 
at $q_1$ and $q_2$.
Since, by assumption, the "runs" $\pi$ and $\pi'$ 
produce different "outputs", the "W-pattern" $P$ 
is "simply divergent", as required.

We now consider the case where $\pi$ and $\pi'$
have large "delay", namely,
$\delay<C m>{\pi,\pi'} > D m^2$.
In this case Lemma \ref{lem:delay-vs-output} guarantees 
the existence of a set
$I \subseteq \{0, 1, \ldots, |\pi|\}$ containing 
$m$ positions in between the input letters such that,
for all pairs $i < j$ in $I$, the interval $[i,j]$
is a "loop" on both $\pi$ and $\pi'$ and satisfies
\begin{equation}\label{equ:pumpedDiff}
  \out{\pump<2>[[i,j]]{\pi}} 
  \:\neq\: 
  \out{\pump<2>[[i,j]]{\pi'}}.
\end{equation}

Next, recall from Equation \eqref{eq:cover-runs}
that $\pi$, and similarly $\pi'$, consists of seven parts,
representing copies of the original "runs" of the "dumbbell"
and consuming the "inputs" $u, v, v, v, v, v, w$.
We identify these parts with the numbers $1,\dots,7$.
Since we defined $m$ as $7E^{H^2 + H + 1}+1$,
there is one of these parts
in which at least $E^{H^2 + H + 1}+1$ of the 
aforementioned positions of $I$ occur.
Let $p \in \{1,2,\ldots,7\}$
denote the number of this part,
and let $I_p$ be a set of $E^{H^2 + H + 1}+1$ positions from $I$
that occur entirely inside the $p$-th part.
We conclude the proof by a further case distinction, 
depending on whether $p \in \{1,7\}$ or $p \in \{2,\dots,6\}$.

\paragraph{Parts $1$ and $7$}
Let us suppose that $p \in \{1,7\}$, 
and let $i$ and $j$ be two distinct positions in $I_p$.
We let $P$ be the "W-pattern" obtained by transforming 
the "dumbbell" as follows:
\begin{enumerate}
    \item First, we "pump" either $\rho_0$ or $\rho_4$
          depending on $p$:
    \begin{itemize}
        \item
        If $p = 1$, we set $\rho_0' = \pump<2>[[i,j]]{\rho_0}$
        and $\rho_4' = \rho_4$.
        \item
        If $p = 7$, we set $\rho_4' = \pump<2>[[i-|u\,v^5|,j-|u\,v^5|]]{\rho_4}$
        and $\rho_0' = \rho_0$.
    \end{itemize}
    \item Then, we add empty "runs" to $\rho_1$, $\rho_2$, and $\rho_3$;
    formally, for each $h \in \{1,2,3\}$, we set
    $\rho_h' = \rho_h$ and $\rho_h'' = \rho_h''' = \varepsilon$.
\end{enumerate}
Now that we identified a "W-pattern" $P$ in $\Cover{T}$, 
we note that 
\[
\begin{aligned}
  \pump<2>[[i,j]]{\pi}  &~=~ \wrun{1,\underline{1},1,1,1} \\
  \pump<2>[[i,j]]{\pi'} &~=~ \wrun{1,1,1,\underline{1},1}.
\end{aligned}
\]
We also recall Equation~\ref{equ:pumpedDiff}, which states 
that these "runs" produce different "outputs".
This means that the "W-pattern" $P$ is "simply divergent".
Finally, since the "runs" of $\Cover{T}$ can be projected 
into "runs" of $T$, we conclude that $T$ contains a 
"simply divergent" "W-pattern".

\paragraph{Parts $2 - 6$}
Let us suppose that $p \in \{2,\dots,6\}$.
Note that, in this case, the elements of $I_p$ 
denote positions inside the $p$-th factor of the "input"
$u\,v\,v\,v\,v\,v\,w$, which is a $v$.
To refer directly to the positions of $v$, we define
$I'_p$ as the set obtained by subtracting $|u\,v^{p-1}|$ 
from each element of $I_p$.
Since the set $|I'_p|$ has cardinality $E^{H^2 + H + 1}+1$,
we claim that we can find an interval with endpoints from $I'_p$ 
that is a "loop" of $\rho_1$, $\rho_2$, and $\rho_3$, at the same time.
Specifically, we can do so via three consecutive applications 
of Lemma~\ref{lem:ramsey}:
\begin{enumerate}
    \item As $|I'_p| = E^{H^2 + H + 1}+1 = E \cdot (E^{H+1})^H + 1$,
    there exists a set $I''_p \subseteq I'_p$ of cardinality $E^{H+1}+1$
    such that for every pair $i < j$ in $I''_p$, the interval $[i,j]$
    is a "loop" of $\rho_1$;
    \item As $|I_p''| = E^{H+1}+1 = E \cdot E^H + 1$,
    there exists $I'''_p \subseteq I''_p$ of cardinality $E+1$
    s.t. for every pair $i < j$ in $I'''_p$, the interval $[i,j]$
    is a "loop" of $\rho_2$ 
    (and also of $\rho_1$, since $i,j \in I'''_p \subseteq I''_p$);
    \item As $|I'''_p| = E + 1 = E \cdot 1^H + 1$, 
    there are two positions $i<j$ in $I'''_p$
    such that the interval $[i,j]$ is a "loop" of $\rho_3$ 
    (and also of $\rho_1$ and $\rho_2$ since $i,j \in I'''_p \subseteq I''_p$).
\end{enumerate}
The diagram below summarizes the current situation:
we have just managed to find an interval $[i,j]$ 
that is a "loop" on all $v$-labelled "runs" 
$\rho_1,\rho_2,\rho_3$ 
of the "dumbbell" (the occurrences of this interval
inside $\rho_1,\rho_2,\rho_3$ are highlighted by thick
segments):
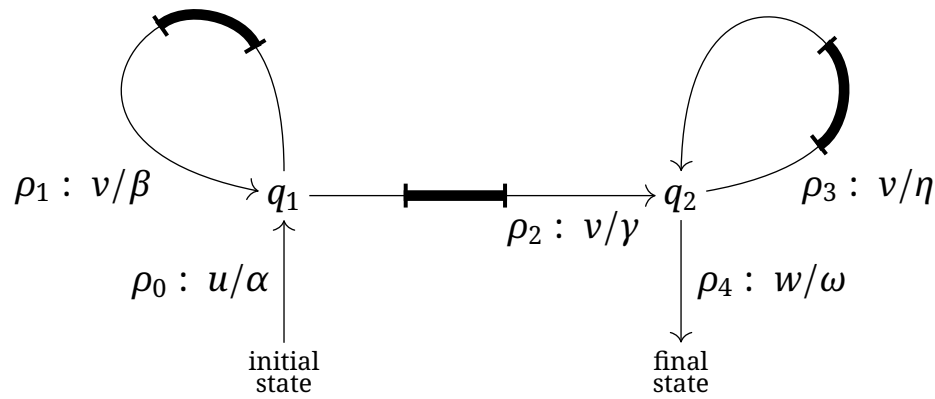
\[
\begin{tikzpicture}[scale=1.5]
\myclip{-2.6,-0.3}{6.1,3.5};
\node (initial) at (0,0.1) 
	  {$\begin{subarray}{c} \text{initial} \\ \text{state} \end{subarray}$};
\node (q) at (0,1.65) {$q_1$};
\node (q') at (3.5,1.65) {$q_2$}; 
\node (final) at (3.5,0.1) 
	  {$\begin{subarray}{c} \text{final} \\ \text{state} \end{subarray}$};
\draw (initial) edge [arrow] 
      node [left=2mm] {$\rho_0:~ u/\alpha$} (q);
\draw (q') edge [arrow] 
      node [right=2mm] {$\rho_4:~ w/\omega$} (final);
\draw (q) edge [arrow,  looseness=40, out=90, in=170]      
	  node [below left, pos=.9] {$\rho_1:~ v/\beta\phantom{xi}$} 
      node [pos=.15] (i) {}
      node [pos=.5] (j) {}
      (q);
\draw (q') edge [arrow,  looseness=40, out=10, in=90]             
	  node [below right, pos=.1] {$\phantom{x}\rho_3:~ v/\eta$} 
      node [pos=.15] (i'') {}
      node [pos=.5] (j'') {}
      (q');
\draw (q) edge [arrow]             
	  node [below = 2mm, pos=.75] {$\rho_2:~ v/\gamma$} 
      node [pos=.25] (i') {}
      node [pos=.6] (j') {}
      (q');
\draw (i) edge [bend right = 55, line width=4] (j);
\draw (i) edge [|-|, bend right = 55, line width=1.5] (j);
\draw (i') edge [line width=4] (j');
\draw (i') edge [|-|, line width=1.5] (j');
\draw (i'') edge [bend right = 55, line width=4] (j'');
\draw (i'') edge [|-|, bend right = 55, line width=1.5] (j'');
\end{tikzpicture}
\]

We can now expose a "W-pattern" $P$ by merging the
positions $i$ and $j$ inside each $v$-labelled "run" 
$\rho_1$, $\rho_2$, and $\rho_3$ of the "dumbbell".
Formally, we let 
$\rho'_0 = \rho_0$, 
$\rho'_4 = \rho_4$, and for every $h \in \{1,2,3\}$,
we define $\rho'_h$, $\rho''_h$, and $\rho'''_h$, respectively,
as the intervals $[0,i]$, $[i,j]$, and $[j,|v|]$ of $\rho_h$.
Now that we have identified a "W-pattern" $P$ inside $\Cover{T}$,
we remark that 
$\pi  = \wrun{1,\underline{1},1,1,1}$ and 
$\pi' = \wrun{1,1,1,\underline{1},1}$.
Additionally, if we transpose $i$ and $j$ from $I_p'$ back to $I$,
that is, if we set $i' = i+|u\,v^{p-1}|$ and $j' = j+|u\,v^{p-1}|$,
since both $i'$ and $j'$ occur in the $p$-th part of $\pi$ and $\pi'$,
"pumping" the interval $[i',j']$ in $\pi$ (resp.~$\pi'$) amounts to 
incrementing the $(p-1)$-th parameter in the notation
$\wrun{1,\underline{1},1,1,1}$ (resp.~$\wrun{1,1,1,\underline{1},1}$).
More precisely:
\[
\begin{aligned}
  \pump<2>[[i',j']]{\pi}   &~=~ \wrun{n_1,\underline{n_2},n_3,n_4,n_5} \\
  \pump<2>[[i',j']]{\pi'}  &~=~ \wrun{n_1,n_2,n_3,\underline{n_4},n_5}
\end{aligned}
\]
where each $n_{p'}$ is either $2$ or $1$ depending on whether $p' = p-1$ or not.
Since Equation~\ref{equ:pumpedDiff} states that these two "runs" 
produce different "outputs", the "W-pattern" $P$ is "simply divergent".
Finally, since the "runs" of $\Cover{T}$ can be projected 
into "runs" of $T$, we conclude that, also in this case, 
$T$ contains a "simply divergent" "W-pattern".
\end{proof}

\section{Conclusion}\label{sec:conclusion}

We have drawn a rather complete
picture of finite-valued SSTs and answered several open questions of~\cite{DBLP:conf/icalp/AlurD11}. 
Regarding expressiveness, finite-valued SSTs can be decomposed as unions of deterministic SSTs
(Theorem~\ref{thm:FiniteValuedToDecomposition}). They are
equivalent to finite-valued two-way transducers
(Theorem~\ref{thm:finval-SST-2FT}), and to "finite-valued"
non-deterministic MSO
transductions (see Section~\ref{sec:intro}). On the algorithmic side,
their equivalence problem is decidable in elementary time
(Theorem~\ref{thm:equiv-finval-SST}) and "finite valuedness" of "SSTs" is
decidable in {\upshape\sc PSpace} ({\upshape\sc PTime} for fixed number of
variables), see Theorem~\ref{thm:effectiveness}. As an alternative proof
to the result of~\cite{YenY22}, our results imply that "finite valuedness" of
two-way transducers can be decided in
{\upshape\sc PSpace} 
(Corollary~\ref{thm:finval-twoway}). Because of the effective expressiveness
equivalence between "SSTs" and non-deterministic MSO transductions, our
result also entails decidability of "finite valuedness" for the latter
class.

\medskip
\subparagraph{\bf Further questions.}
A first natural question is how big the valuedness of an "SST" can
be.
In the classical case of one-way transducers the valuedness has been
shown to be at most exponential (if finite)~\cite{Weber93}.
We can obtain a bound from Lemma~\ref{lem:unbounded-values3}, but the
value is likely to be sub-optimal.

Our "equivalence" procedure relies on the
decomposition of a $k$-valued "SST" into a union of $k$ "deterministic"
"SSTs" each of elementary size. The latter construction is likely to
be sub-optimal, too, and so is our complexity for checking "equivalence". 
On the other hand, only a {\upshape\sc PSpace} lower bound is
known, which follows easily from a reduction of NFA equivalence~\cite{AhoHU74}. 
A better understanding of the
complexity of the "equivalence" problem for (sub)classes of "SSTs" is a
challenging question. Already for "deterministic" "SSTs", the complexity of the "equivalence" problem
is only known to lie between {\upshape\sc NLogSpace} and
{\upshape\sc PSpace}~\cite{DBLP:conf/popl/AlurC11}.

However, beyond the "finite-valued" setting there is little hope to 
find a natural restriction on "valuedness" which would preserve 
the decidability of the "equivalence" problem. Already for 
one-way transducers of linear "valuedness" (i.e.~where the number 
of outputs is linear in the "input" length), "equivalence" is
undecidable, as shown through a small modification of the 
proof of~\cite{iba78siam}.

"Deterministic" "SSTs" have been extended,
while preserving decidability of the "equivalence problem",
in several ways:
to copyful "SSTs"~\cite{DBLP:journals/siglog/Bojanczyk19a,DBLP:journals/fuin/FiliotR21},
which allow to copy the content of variables several times, 
to infinite strings~\cite{DBLP:conf/lics/AlurFT12}, and to
trees~\cite{DBLP:journals/jacm/AlurD17}. 
Generalizations of these results to the "finite-valued" setting 
yield interesting questions. On trees, similar questions 
(effective "finite valuedness", decomposition and "equivalence") 
have been answered positively for bottom-up tree
transducers~\cite{DBLP:journals/mst/Seidl94}.

Finally, "SSTs" have linear "input"-to-"output" growth (in the length of the
strings). There is a recent trend in extending transducer models 
to allow polynomial growth
\cite{DBLP:conf/icalp/BojanczykKL19,DBLP:conf/lics/Bojanczyk22,DBLP:conf/lics/Bojanczyk23a,GaetanPhD}.
"Finite valuedness" has not  been studied in this context yet.

    \printbibliography


\end{document}